\documentclass[oneside,11pt]{article}

\usepackage[utf8]{inputenc}
\usepackage[T1]{fontenc}
\usepackage{enumerate}
\usepackage{amsthm}
\usepackage{amsmath}
\usepackage{amsfonts}
\usepackage{amssymb}

\usepackage{graphicx} 
\usepackage{float}
\usepackage{booktabs}
\usepackage{multirow}
\usepackage{rotating}
\usepackage{array}
\usepackage{xcolor}
\usepackage{subfig}
\usepackage[ruled]{algorithm2e}

\newcolumntype{R}[1]{>{\raggedleft\let\newline\\\arraybackslash\hspace{0pt}}m{#1}}
\newcolumntype{L}[1]{>{\raggedright\let\newline\\\arraybackslash\hspace{0pt}}m{#1}}

\usepackage{breakcites}

\usepackage{nowidow}

\usepackage[top=1.5cm,bottom=2cm,right=2.5cm,left=2.5cm]{geometry}
\usepackage{titling}

\usepackage{natbib}

\usepackage{ifplatform}

\usepackage{textcomp}

\ifwindows
\fi

\allowdisplaybreaks

\usepackage{comment}

\newcommand{\bm}{\mathbf{m}}

\newcommand{\norm}[1]{\left|\left| #1\right|\right|}

\DeclareFontFamily{OT1}{pzc}{}
\DeclareFontShape{OT1}{pzc}{m}{it}{<-> s * [1.10] pzcmi7t}{}
\DeclareMathAlphabet{\mathpzc}{OT1}{pzc}{m}{it}

\newtheorem{theorem}{Theorem}
\newtheorem{corollary}{Corollary}
\newtheorem{remark}{Remark}

\newtheorem{lemma}{Lemma}

\usepackage[pdftex,final,bookmarksnumbered,bookmarksopen=false,breaklinks,colorlinks]{hyperref}
\hypersetup{
	final,
	bookmarksnumbered=true,
	bookmarksopen=true,
	bookmarksopenlevel=0,
	unicode=false,          
	pdftoolbar=true,        
	pdfmenubar=true,        
	pdffitwindow=false,     
	pdftitle={A flexible functional-circular regression model for analyzing temperature curves},    
	pdfdisplaydoctitle=true, 
	pdftoolbar=true,
	pdfmenubar=true,
	pdfcreator={},   
	pdfproducer={}, 
	pdfkeywords={Directional data}{Circular data}{Nonparametric statistics}{Smoothing}, 
	pdfnewwindow=true,      
	breaklinks=true,
	hidelinks,       
	linkcolor=black,          
	citecolor=black,        
	filecolor=magenta,      
	urlcolor=cyan           
}
\begin{document}

\predate{}%
\postdate{}%
\date{}
\title{\LARGE A flexible functional-circular regression model for analyzing temperature curves\\}\normalsize
\author{Andrea Meil\'an-Vila$^{ \small 1,*}$ \\
	Carlos III University of Madrid
	\and Rosa M. Crujeiras$^{\small 2}$ \\
	University of Santiago de Compostela\\
	\and %
	Mario Francisco-Fern\'andez$^{ \small 3}$\\
	University of  A Coru\~na}
\maketitle


%
\footnotetext[1]{Department of Statistics, Carlos III University of Madrid, Av. de la Universidad 30, Leganés, 28911, Spain, $^{*}$ameilan@est-econ.uc3m.es}
\footnotetext[2]{CITMAga, Galician Centre for Mathematical Research and Technology, University of Santiago de Compostela, San
	Xerome s/n,Santiago de Compostela, 15782, Spain.}
\footnotetext[3]{Department of Mathematics, University of  A Coru\~na. CITIC, Faculty of Computer Science, Campus de Elviña
	s/n, A Coruña, 15071, Spain.}


\begin{abstract} 
Changes on temperature patterns, on a local scale, are perceived by individuals as the most direct indicators of global warming and climate change. As a specific example, for an Atlantic climate location, {spring and fall seasons} should present a mild transition between {winter} and {summer, and summer and winter, respectively}. By observing daily temperature curves along time, being each curve attached to a certain calendar day, a regression model for these variables (temperature curve as covariate and calendar day as response) would be useful for modeling their relation for a certain period. In addition, temperature changes could be assessed by prediction and observation comparisons in the long run. Such a model is presented and studied in this work, considering a {nonparametric} Nadaraya-Watson-type estimator for functional covariate and circular response. The asymptotic bias and variance of this estimator, as {well} as its asymptotic distribution are derived. Its finite sample performance is evaluated in a simulation study and the proposal is applied to investigate a real-data set concerning temperature curves.
\end{abstract}
\begin{flushleft}
\end{flushleft}

\section{Introduction}
\label{intro}

  {The {State of the Climate} report by \cite{Blunden2020} corresponding to 2019, includes a series of analysis of fully monitored variables on a global scale. One of {these} variables is surface temperature, revealing that July 2019 was Earth's hottest month on record. In Europe, 2019 was the second hottest year (following 2018), and 2014-2019 are Europe's warmest years on record. The report indicates that the warming of land and ocean surfaces is reflected across the Globe, with lakes and permafrost temperatures increasing, as an evidence of climate change.}

{However, as pointed out by \cite{Bloodhart2015}, it is difficult for individuals to determine if they have experienced the effects of climate change, given that their information refers to a reduced period of time and is usually restricted to a local scale. Nevertheless, understanding how people experience the changes on local weather patterns is important, given that personal experiences are known to affect climate change beliefs \citep{Goebbert2012} and risk perceptions. {Hence, they may condition citizen's support on prevention policies \citep[see][among others]{Howe2018, Goebbert2012, Taylor2014, Bloodhart2015}}. {As pointed out by \cite{Goebbert2012} on a global scale, personal perceptions on weather patterns are usually noticed through temperature changes and this indicator is usually taken as an evidence by individuals for making inferences about climate changes. However, these subjective insights should be confirmed with the construction of an appropriate statistical regression model that allows for assessing if relevant changes in temperature patterns could have happened in different periods of time.}}

{As a motivating example, consider daily temperature records in Santiago de Compostela (NW-Spain) for the period {2002}-2019. Temperature curves from February 15, 2002 until June 28, 2005 can be seen in Figure~\ref{figure:rd} (left), where the color scale {indicates} the day of the year when each curve was observed (note that the scale-palette preserves the periodicity of the data). In Figure~\ref{figure:rd} (right), a functional boxplot {\citep{Sun2011}} is presented, plotting the 50\% central region for the observed temperature curves. {The vertical segments are the whiskers of the boxplot}. {Modeling} the relation between the temperature curve (as a functional covariate) and the day of the year {basis} (as a circular response) for a certain period of time would allow to investigate, by comparing {the predicted day by such a model for a given temperature curve, in a di\-ffe\-rent period, with the actual day corresponding to that temperature curve,}
	if temperature curve patterns are stable along time or if, on the contrary, observed temperature curves are displaced.}

\begin{figure}[htb]
	\centering
	\includegraphics[width=01\textwidth]{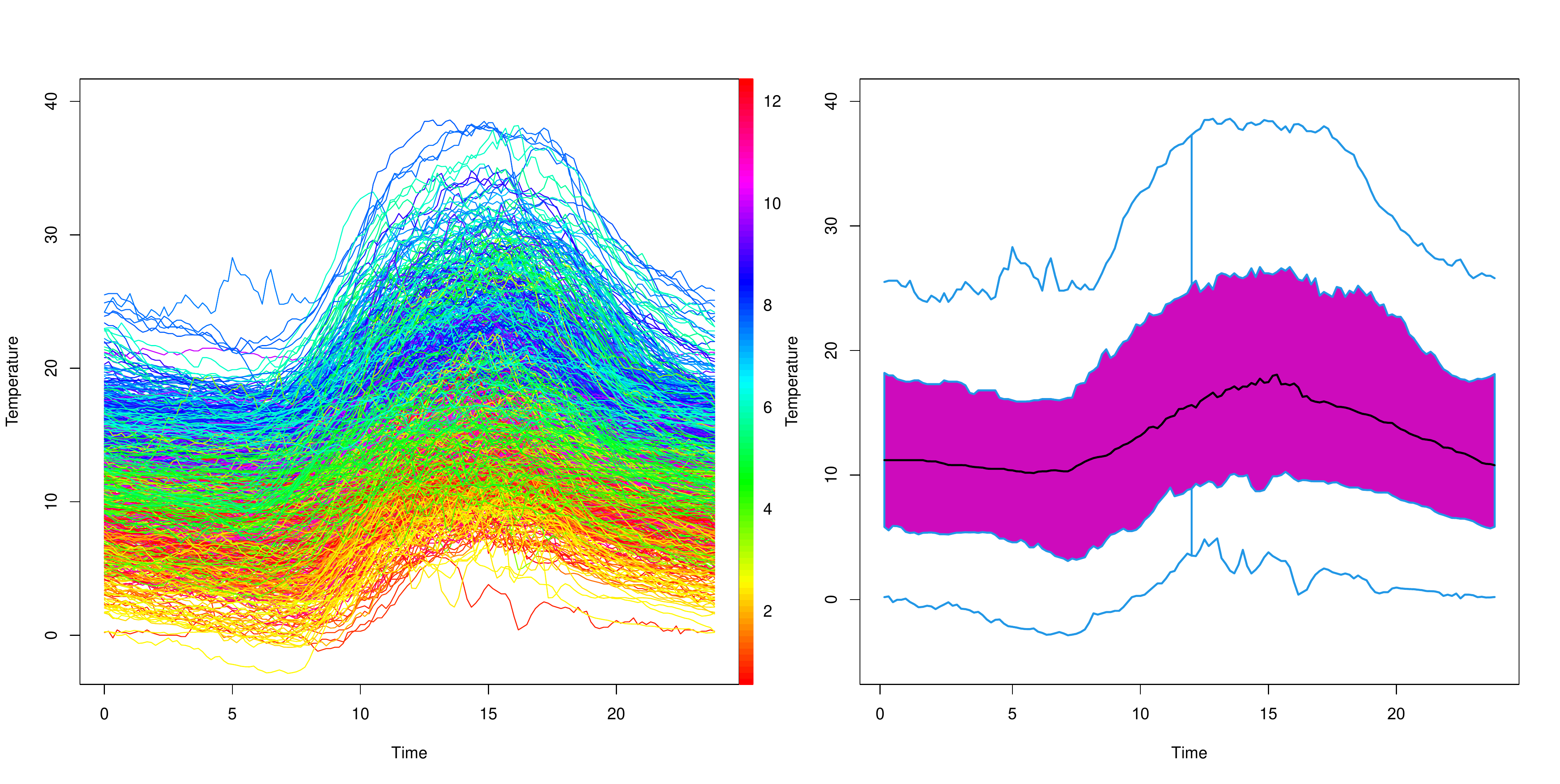}	
	\caption{Daily temperature curves
		in Santiago de Compostela (Spain) from February 15, 2002 to June 28, 2005 (left) and its corresponding functional boxplot (right).}
	\label{figure:rd}
\end{figure}


{The novel methodology presented in this paper exploits the usage of high frequency temperature records, {considering a (flexible) regression model with a functional covariate (the whole daily temperature curve) and a circular response (a day on a year basis)}. Circular data can be viewed as points on the unit circle, so fixing a direction and a sense of rotation, circular data can be expressed as angles. For a complete introduction on circular data, we refer to \citet{mardia2000directional} or \citet{ley2017modern}. This type of data appear in different applied fields, and  some examples include
	wind directions \citep{fisher1995statistical},  {wave} directions \citep{jona2012spatial,wang2015joint}, or animal orientations \citep{scapini2002multiple},  among others.   {Regression analysis considering models where the response and/or the covariables are circular variables has been addressed in different papers.}
	Parametric approaches were considered in \cite{fisher1992regression} and \cite{presnell1998projected} for regression models with circular response and Euclidean covariates. The authors assumed a parametric  distribution model for the circular response variable. An alternative parametric regression model was also analyzed in \cite{kim2017multivariate}, considering that a multivariate circular variable depends on several circular covariates. Using nonparametric methods, \cite{di2013non} introduced a  regression estimator  for models with circular response, a single real-valued covariate and independent errors and also for errors coming from mixing processes. A similar approach   was also applied for time series in \cite{di2012non}.  {Following similar ideas,  \cite{meilan2019c} proposed and studied local polynomial-type regression estimators considering { a model with} a circular response and several real-valued covariates for independent data}. This class of estimators was also analyzed in \cite{meilan2020a}  in the presence of spatial correlation.}

{ {A natural extension of the multiple regression model used in \cite{meilan2019c} would be to consider a functional covariate belonging to an infinite-dimensional space. Nowadays, with the advances in data collection methods, more and more data are being recorded during an interval time or at several discrete time points with a high frequency, producing functional data. For example, in many fields of applied sciences such as chemometrics \citep{abraham2003unsupervised}, biometrics \citep{gasser1998nonparametric} and medicine \citep{antoniadis2007estimation}, among others, it is quite common to deal with observations that are curves or functions. The branch of statistics that analyzes this type of data is called functional data analysis \citep[see, for instance,][]{ramsay2005functional,ferraty2006nonparametric,Kokoszka_17}.}
	The  literature on functional regression modeling  is really extensive \citep[see{, for example,}][{for complete reviews and recent advances in this field}]{greven2017general,morris2015functional,Aneiros_19,Aneiros_22,Goia_16}{, including parametric (in particular, linear) and nonparametric models}. Assuming parametric conditions, earlier advances on functional regression were introduced by \cite{ramsay2005functional}, while  a recent overview was provided in \cite{febrero2017functional}. The nonparametric  {methodologies} became popular in the functional regression context with the book by \cite{ferraty2006nonparametric},  {this topic being}  widely  {addressed}
	in the last decade \citep[see][for a survey]{ling2018nonparametric}.  {In this framework, if the explanatory variables are functional,}
	nonparametric regression approaches are essentially based on an  {adaptation} of the 
	{Nadaraya--Watson}   \citep{ferraty2002functional,ferraty2006nonparametric,ling2020convergence} or the local linear regression  {estimators} \citep{aneiros2011functional,berlinet_11,boj2008local,baillo2009local,ferraty2022scalar}. 
	
	{ {As pointed our before, t}he aim of the present work is to propose and study a nonparametric regression estimator for a model with a circular response and a functional covariate.  When the response variable is circular, the  regression function   can be defined as the minimizer of
		a circular risk function. It can be proved  that the minimizer of this risk function is the inverse tangent function of the ratio between  the conditional expectation of the sine and the cosine of the response variable.  The proposal introduced in this work implicitly con\-si\-ders  {two regression models, one for the sine and another for the cosine of the response variable}. Then, a nonparametric estimator for the regression function is directly obtained by calculating the inverse tangent function of the ratio of appropriate  Nadaraya--Watson estimators for the two  regression functions of the sine and cosine models. This way of proceeding has already been con\-si\-de\-red  {previously.} For instance, a similar approach was employed in \cite{meilan2019c} and \cite{meilan2020a} for regression models with a circular response, Euclidean covariates and independent or spatially correlated errors, respectively. As in any nonparametric approach, a crucial issue  is the selection of an appropriate bandwidth or smoothing parameter. In this paper, cross-validation bandwidth selection methods adapted for the current framework are introduced and analyzed in practice. }

	{This paper is organized as follows. In Section \ref{sec:model}, the regression model for a circular response and  {a functional covariate}  is presented. In Section \ref{sec:est}, the nonparametric estimator of the circular regression function is proposed. Expressions for its asymptotic bias and variance, as well as its asymptotic distribution  are included  in Section \ref{sec:asyres}. 
		The finite sample performance of the estimator is studied through simulations in Section \ref{sec:sim}, and  illustrated in Section \ref{sec:example} with  {the} real data example on temperature curves. Finally, a discussion on this paper is included in Section \ref{sec:dis}.}
	
	\section{Functional-circular regression model}\label{sec:model}
	Let $\{({\mathcal{X}}_i,\Theta_i)\}_{i=1}^n$ be a random sample from the  random vector (${\mathcal{X}}, \Theta$) {, where $\Theta$ is a circular random variable taking values on $\mathbb{T}=[0,2\pi)$, and ${\mathcal{X}}$ is a functional variable supported on $E$, a separable Banach space endowed with a norm $\norm{\cdot}$.}  This general framework includes $L^p$, Sobolev and Besov spaces. Separability condition avoids measurability problems for the random variable $\mathcal{X}$.
	
	Assume that the circular random variable $\Theta$  depends on the functional random variable ${\mathcal{X}}$ through the following  regression model:
	\begin{equation}\label{model}
		\Theta_i=[m(\mathcal{X}_i)+\varepsilon_i](\mbox{\texttt{mod}} \, 2\pi), \quad i=1,\dots,n,
	\end{equation}
	where $m$ is a smooth trend or regression function mapping $E$ onto $\mathbb{T}$, \texttt{mod} stands for the modulo operation and $\varepsilon_i, i=1,\ldots,n$, is an independent sample of a circular {random} variable $\varepsilon$, satisfying ${\rm E}[\sin (\varepsilon)\mid\mathcal{X}=\chi]=0$ and having finite concentration. In this setting, the circular regression function $m$ in model (\ref{model}) can be defined as the minimizer of the risk function ${\rm E}\{1-\cos[\Theta-m(\mathcal{X})]\mid \mathcal{X}=\chi\}$. The minimizer of this cosine risk is given by: \begin{equation}m(\chi)=\mbox{atan2}[m_1(\chi),m_2(\chi)],
		\label{m}
	\end{equation}
	where $m_1(\chi)={\rm E}[\sin(\Theta)\mid\mathcal{X}=\chi]$, $m_2(\chi)={\rm E}[\cos(\Theta)\mid\mathcal{X}=\chi]$ and the function $\mbox{atan2}(y,x)$ returns the angle between the
	$x$-axis and the vector from the origin to $(x, y)$. With this formulation, $m_1$ and $m_2$ can be  {considered} as the regression functions of two regression models,  {being} $\sin(\Theta)$ and $\cos(\Theta)$ their  {response variables}, respectively, and with functional covariates.
	Specifically, we assume the models
	\begin{equation}
		\sin(\Theta_i)=m_1(\mathcal{X}_i)+\xi_i,\quad i=1,\dots,n,\label{model1}
	\end{equation}
	and
	\begin{equation}
		\cos(\Theta_i)=m_2(\mathcal{X}_i)+\zeta_i,\quad i=1,\dots,n,\label{model2}
	\end{equation}
	where $m_1$ and $m_2$ are regression functions mapping $E$ onto $[-1,1]$. The $\xi_i$ and the $\zeta_i$  are independent error terms, absolutely bounded by 1,  satisfying ${\rm E}(\xi\mid\mathcal{X}=\chi)={\rm E}(\zeta \mid\mathcal{X}=\chi)=0$. Additionally, for every $\chi\in E$, set $s_1^2(\chi)={\rm Var}(\xi\mid\mathcal{X}=\chi)$, $s_2^2(\chi)={\rm Var}(\zeta\mid\mathcal{X}=\chi)$ and $c(\chi)={\rm E}(\xi\zeta\mid\mathcal{X}=\chi)$.
	The assumption that  {models (\ref{model1}) and (\ref{model2}) simultaneously hold with (\ref{model}) leads to certain relations between the variances and covariances of the errors in these models}, as  {it is} described below.
	
	Set $\ell(\chi) ={\rm E}[\cos (\varepsilon)\mid\mathcal{X}=\chi]$,  $\sigma^2_1(\chi)={\rm Var}[\sin(\varepsilon)\mid\mathcal{X}=\chi]$,
	$\sigma^2_2(\chi)={\rm Var}[\cos(\varepsilon)\mid\mathcal{X}=\chi]$ and
	$\sigma_{12}(\chi)={\rm E}[\sin(\varepsilon)\cos(\varepsilon)\mid\mathcal{X}=\chi]$.
	Then, using the sine and cosine addition formulas
	in model (\ref{model}), it follows that, for $i=1, \ldots, n$,
	\begin{equation}
		\sin(\Theta_i)=\sin[m(\mathcal{X}_i)] \cos(\varepsilon_i)+\cos[m(\mathcal{X}_i)] \sin(\varepsilon_i)
		\label{eq_sin}
	\end{equation}
	and
	\begin{equation}
		\cos(\Theta_i)=\cos[m(\mathcal{X}_i)] \cos(\varepsilon_i)-\sin[m(\mathcal{X}_i)] \sin(\varepsilon_i).
		\label{eq_cos}
	\end{equation}
	Hence, defining
	$f_1(\chi)=\sin [m(\chi)]$ and $f_2(\chi)=\cos [m(\chi)]$ and applying conditional expectations in (\ref{eq_sin}) and (\ref{eq_cos}), it holds that
	\begin{eqnarray}\label{m1m2}
		m_1(\chi)=f_1(\chi)\ell(\chi)\quad\mbox{and}\quad
		m_2(\chi)=f_2(\chi)\ell(\chi).
	\end{eqnarray}Note that $f_1(\chi)$ and $f_2(\chi)$ correspond to the {normalized versions} of $m_1(\chi)$ and $m_2(\chi)$, respectively. Indeed, taking into account that  $f_1^2(\chi)+f_2^2(\chi)=1$, it can be easily deduced that $\ell(\chi)=[m^2_1(\chi)+m_2^2(\chi)]^{1/2}$. Hence,  under model $(\ref{model})$, $\ell(\chi)$ amounts to the mean resultant length of $\Theta$ given $\mathcal{X}=\chi$,  {and taking into account that it is assumed that ${\rm E}[\sin(\varepsilon)\mid \mathcal{X}=\chi]=0$}, also corresponds to the mean resultant length of $\varepsilon$ given $\mathcal{X}=\chi$. The following explicit expressions for the conditional variances of the error terms involved in models  (\ref{model1}) and (\ref{model2}), in terms of the conditional variances and covariance of the Cartesian coordinates of $\varepsilon$, can be obtained:
	\begin{equation}
		s_1^2(\chi)=f_1^2(\chi)\sigma^2_2(\chi)+2f_1(\chi)f_2(\chi)\sigma_{12}(\chi)+f_2^2(\chi)\sigma^2_1(\chi),
		\label{s12}
	\end{equation}
	\vspace*{-0.5cm}
	\begin{equation}
		s_2^2(\chi)=f_2^2(\chi)\sigma^2_2(\chi)-2f_2(\chi)f_1(\chi)\sigma_{12}(\chi)+f_1^2(\chi)\sigma^2_1(\chi),
		\label{s22}
	\end{equation}
	as well as for the covariance between the error terms in (\ref{model1}) and (\ref{model2}),
	\begin{eqnarray}
		\label{c}
		\nonumber c(\chi)&=&f_1(\chi)f_2(\chi)\sigma^2_2(\chi)-f_1^2(\chi)\sigma_{12}(\chi)\\
		&+&f_2^2(\chi)\sigma_{12}(\chi)-f_1(\chi)f_2(\chi)\sigma^2_1(\chi).
	\end{eqnarray}

	{A nonparametric Nadaraya--Watson-type estimator}  {of the circular re\-gre\-ssion function $m$ in model (\ref{model}) is presented and studied}   {in what follows.}

	\section{Nonparametric regression estimator}\label{sec:est}
	{A Nadaraya--Watson-type estimator for {$m(\chi)$}, in model (\ref{model}) can be defined by replacing {$m_1(\chi)$} and {$m_2(\chi)$} in (\ref{m}) by suitable Nadaraya--Watson estimators. Specifically, the following estimator:
		\begin{equation}\label{C_est}
			\hat{m}_{h}(\chi)={\rm atan2}[\hat{m}_{1, h}(\chi),\hat{m}_{2, h}(\chi)]
		\end{equation}
		is considered, where $\hat m_{1, h}(\chi)$ and $\hat m_{2, h}(\chi)$ denote  the Nadaraya--Watson estimators  of $ m_1(\chi)$ and $m_2(\chi)$, respectively.}
	The asymptotic  {properties (bias, variance and asymptotic normality)} of estimator (\ref{C_est}) are derived  {in the following section}.  {For this purpose, assuming that models (\ref{model1}) and (\ref{model2}) hold, the} asymptotic properties of  {$\hat m_{1, h}(\chi)$ and $\hat m_{2, h}(\chi)$} are used.\\
	
	Considering models  (\ref{model1}) and (\ref{model2}),  Nadaraya--Watson estimators for {$m_j(\chi)$}, $j=1,2$, are respectively defined as:
	\begin{equation}
		\label{C_estNadaraya--Watson}\hat m_{j, h}(\chi)=\left\{\begin{array}{lc}\dfrac{\sum_{i=1}^n K(h^{-1}\norm{\mathcal{X}_i-\chi})\sin(\Theta_i)}{\sum_{i=1}^n K(h^{-1}\norm{\mathcal{X}_i-\chi})}&\text{if $j=1$},\\\\ \dfrac{\sum_{i=1}^n K(h^{-1}\norm{\mathcal{X}_i-\chi})\cos(\Theta_i)}{\sum_{i=1}^n K(h^{-1}\norm{\mathcal{X}_i-\chi})}&\text{if $j=2$},\end{array}\right.
	\end{equation}
	where $K$ is a symmetric kernel and $h=h(n)$ is a strictly positive real-valued bandwidth{ controlling the smoothness of the estimator. Note that in this functional framework,} the support of $K$ should be contained in $\mathbb R^+$, since $ {\norm{\chi-\chi'}}\ge 0$, for all $\chi,\chi'\in E$.
	{If the kernel $K$ is positive with support on $[0,1]$, the estimators given in (\ref{C_estNadaraya--Watson}) only consider the observations $\sin(\Theta_i)$ and $\cos(\Theta_i)$, respectively,  associated to the curves $\mathcal{X}_i$ such that $\norm{\mathcal{X}_i-\chi}\le {h}$, since $K(h^{-1}\norm{\mathcal{X}_i-\chi})=0$ when the distance between $\chi$ and $\mathcal{X}_i$ is larger than $h$.} The regression  {estimators} given in (\ref{C_estNadaraya--Watson})  {are the adaptations} to the functional context of the  {classical finite-dimensional} Nadaraya--Watson  {estimators}   for the sine and cosine regression models  {\citep{ferraty2007nonparametric}}.
	
	{Although the choice of the kernel function is of secondary importance, the bandwidth parameter plays an important role in the performance of the Nadaraya--Watson estimators \eqref{C_estNadaraya--Watson} and, consequently, of the regression estimator \eqref{C_est}. When $h$ is  {excessively large}, the number of observations involved in the estimation method  {will also be too large}, and   {an} oversmoothed estimator is obtained. Conversely, if the bandwidth parameter is too small, few observations  {will be} considered in the estimation procedure, and a undersmoothed estimator is computed.
		Therefore, in this case, as in any other kernel-based estimator, in practice, data-driven bandwidth selection methods are
		needed.
		A cross-validation approach is used to select the bandwidth parameter $h$ for \eqref{C_est} in the   {simulation study and for the real data application.} {Note that the same bandwidth is {used} for computing $\hat m_{1, h}$ and $\hat m_{2, h}$ in (\ref{C_estNadaraya--Watson}). If two different bandwidths for sine and cosine components were considered, this would imply that possibly different curves were chosen for computing regression estimators of sine and cosine models, which seems quite unnatural, since the cartesian coordinates are directly related with the angle itself.}
	}

	\section{Asymptotic results}\label{sec:asyres}
	In this section, the asymptotic  bias and variance expressions for $\hat{m}_{h}(\chi)$, given in (\ref{C_est}), are derived.  Moreover, the asymptotic distribution of the proposed estimator is calculated.
	{The proofs of all these theoretical results are collected in {Appendix \ref{proofs}}.}
	
	\subsection{Asymptotic bias and variance}\label{sec:asybiasvar}
	
	{To derive} the asymptotic bias and variance of the estimator given in (\ref{C_est}), the asymptotic properties of the  Nadaraya--Watson  estimators of $m_j(\chi)$, $j=1,2,$  {given in \eqref{C_estNadaraya--Watson},} are required.  {Using the asymptotic results given in \cite{ferraty2007nonparametric}}, the asymptotic bias and  variance  of $\hat m_{j, h}(\chi)$, $j=1,2$, are  {immediately obtained}.  These   {expressions, jointly} with the covariance between $\hat m_{1, h}(\chi)$ and $\hat m_{2, h}(\chi)$, are  collected in Lemma \ref{C_pro1}. 
	
	Let $\varphi_{\chi}$ and $\varphi_{j,\chi}$, $j=1,2$, be the functions defined for all $s\in\mathbb{R}$ by:
	\begin{eqnarray}
		\varphi_{\chi}(s)=\mathbb{E}\left\{[m(\mathcal{X})-m(\chi)]\mid \norm{\mathcal{X}-\chi}=s\right\},
		\label{varphi}\end{eqnarray}
	\begin{eqnarray}\varphi_{j,\chi}(s)=\mathbb{E}\left\{[m_j(\mathcal{X})-m_j(\chi)]\mid \norm{\mathcal{X}-\chi}=s\right\},
		\label{varphi12}\end{eqnarray}
	and denote by  {$F_{\chi}$} the cumulative  {distribution} function  {of} the random variable $\norm{\mathcal{X}-\chi}$,
	$$F_{\chi}(t)=\mathbb{P}(\norm{\mathcal{X}-\chi}\le t), \quad t\in\mathbb R.$$
	In addition, denoting,  {$\forall s\in[0,1],$}
	$$\tau_{\chi,h}(s)=\dfrac{F_{\chi}(hs)}{F_{\chi}(h)}=\mathbb{P}(\norm{\mathcal{X}-\chi}\le hs\mid \norm{\mathcal{X}-\chi}\le h),$$
	the following assumptions are required. 
	\begin{enumerate}[{(C}1)]
		\item The regression functions $m_j$ and  $s_j^2$, for $j=1,2$,  are continuous in a neighborhood of $\chi\in E$, and $F_{\chi}(0)=0$.
		\item  $\varphi'_{j,\chi}(0)$ {, for $j=1,2,$} exist.
		\item The bandwidth $h$ satisfies $h\to 0$  and $nF_{\chi}(h)\to\infty$, as $n\to\infty$.
		\item  The kernel $K$ is supported on $[0,1]$ and has a continuous derivative on $[0,1)$. Moreover, $K'(s)\le 0$ and $K(1)>0$.
		\item For all $s\in [0,1]$, $\tau_{\chi,h}(s)\to \tau_{\chi,0}(s)$, when $h\to 0$.
	\end{enumerate}
	
	Assumptions (C1), (C3)--(C4) are similar  {to} those classically employed in the finite-dimensional context. Regarding   (C2), as pointed out by \cite{ferraty2007nonparametric}, this condition avoids  {the difficulties} that could appear  when con\-si\-de\-ring differential calculus on Banach spaces.  Thus,  regularity hypothesis on the regression function $m$, such as being twice continuously differentiable, are not required. In assumption (C5), the function $\tau_{\chi,0}$  {is included in the constants that are part of the main terms of the asymptotic expansions of the bias and the variance of the regression estimator \eqref{C_est}. The close expression of this function in some particular cases are provided in \citet[][Proposition 1]{ferraty2007nonparametric}.}
	
	\begin{lemma}
		\label{C_pro1}
		Let $\{(\mathcal{X}_i,\Theta_i)\}_{i=1}^n$ be a random sample  from  {(${\mathcal{X}}, \Theta$)} supported on $E\times \mathbb{T}$.  Under assumptions  \textnormal{(C1)--(C5)}, for $\chi\in E$, then, for $j=1,2$, 
		\begin{eqnarray*}
			\label{C_expp0}
			{\mathbb{E}}[\hat m_{j, h}(\chi)- m_j(\chi)]&=&\varphi_{j,\chi}'(0)\dfrac{M_{\chi,0}}{M_{\chi,1}}h+\mathcal{O}\bigg[\frac{1}{nF_{\chi}(h)}\bigg]
			+\mathpzc{o}(h),\\
			\label{C_varp0}
			{\mathbb{V}{\rm ar}}[\hat m_{j, h}(\chi)]&=&\dfrac{s_j^2(\chi)}{nF_{\chi}(h)}\dfrac{M_{\chi,2}}{M^2_{\chi,1}}+\mathpzc{o}\bigg[\frac{1}{nF_{\chi}(h)}\bigg],\\
			\mathbb{C}{\rm ov}[\hat{m}_{1, h}(\chi),\hat{m}_{2, h}(\chi)]&=&\dfrac{c(\chi)}{nF_{\chi}(h)}\dfrac{M_{\chi,2}}{M^2_{\chi,1}}+\mathpzc{o}\bigg[\frac{1}{nF_{\chi}(h)}\bigg],\label{cov_m1m2_circular_NW}
		\end{eqnarray*}
		where 
		$$
		M_{\chi,0}=K(1)-\int_0^1[sK(s)]'\tau_{\chi,0}(s)ds,
		$$
		$$
		M_{\chi,1}=K(1)-\int_0^1K'(s)\tau_{\chi,0}(s)ds
		$$
		and
		$$
		M_{\chi,2}=K^2(1)-\int_0^1(K^2)'(s)\tau_{\chi,0}(s)ds.
		$$
	\end{lemma}

	\begin{remark}
		{In the case of regression models with Euclidean response and multiple covariates, the main term of the asymptotic bias of the Nadaraya--Watson estimator depends on the gradient and the Hessian matrix of the regression functions \citep[see][for further details]{hardle_muller}. Nevertheless, in the present functional setting, the leading term of the asymptotic bias of $\hat m_{j, h}$ depends on the first derivative of the functions $\varphi_{j,\chi}$ evaluated at zero.}
	\end{remark}

	Now, using the previous lemma, the following theorem provides the asymptotic bias and variance of the estimator $\hat m_{h}(\chi)$.
	
	\begin{theorem}\label{C_teoNadaraya--Watson}
		Let $\{(\mathcal{X}_i,\Theta_i)\}_{i=1}^n$ be a random sample  from  {(${\mathcal{X}}, \Theta$)} supported on $E\times \mathbb{T}$. Under assumptions  \textnormal{(C1)--(C5)}, the asymptotic  bias of estimator  $\hat{m}_{h}(\chi)$, for  $\chi\in E$, is given by:
		\begin{eqnarray}\label{bias_circular_NW}
			{\mathbb{E}}[\hat{m}_{h}(\chi)-m(\chi)]&=&\varphi_{\chi}'(0)\dfrac{M_{\chi,0}}{M_{\chi,1}}h+\mathcal{O}\bigg[\frac{1}{nF_{\chi}(h)}\bigg]
			+\mathpzc{o}(h),
		\end{eqnarray}
		and the asymptotic variance is:
		\begin{eqnarray}\label{variance_circular_NW}
			{\mathbb{V}{\rm ar}}[\hat{m}_{h}(\chi)]=\dfrac{1}{nF_{\chi}(h)} \dfrac{\sigma^2_1(\chi)}{\ell^2(\chi)}\dfrac{M_{\chi,2}}{M^2_{\chi,1}}+\mathpzc{o}\bigg[\frac{1}{nF_{\chi}(h)}\bigg].
		\end{eqnarray}
	\end{theorem}
	
	{The asymptotic expressions of the bias and variance obtained in Theorem~\ref{C_teoNadaraya--Watson} can be used to derive the asymptotic mean squared error (AMSE) of $\hat{m}_{h}(\chi)$. An asymptotically optimal local bandwidth for $\hat{m}_{h}(\chi)$  {can be} selected mi\-ni\-mi\-zing this error criterion with respect to $h$. A plug-in bandwidth selector could be defined replacing the unknown quantities in the asymptotic optimal parameter by appropriate estimates. This problem is more complex in this functional context than in the finite-dimensional setting \citep{ferraty2007nonparametric}. Moreover, the design of global plug-in selectors would require obtaining integrated versions of the AMSE. Making this issue rigorous poses challenges that continue to be topics of current research. Therefore, in this paper, we avoid using plug-in bandwidth selectors and, as pointed out before, {we} employ a suitable adapted cross-validation technique to select the bandwidth for $\hat{m}_{h}$ in the numerical studies (simulations and real data application).}

	\subsection{Asymptotic normality}\label{sec:asynorm}
	
	Next, the asymptotic distribution of the proposed nonparametric regression estimator  {\eqref{C_est}} is derived. 
	Apart from  the  {assumptions} stated in Section \ref{sec:asybiasvar} for computing its asymptotic bias and variance, the following  {condition}  is needed:
	
	\begin{enumerate}[{(C}6)]
		\item  $\varphi'_{\chi}(0)\neq 0$ and $M_{\chi,0}$ in Lemma \ref{C_pro1} is strictly positive.
	\end{enumerate}
	
	Condition (C6) is required to ensure that the leading term of the  bias does not vanish. The first part of this assumption is similar to the one stated in the finite-dimensional
	setting {. The second part of (C6) is specific for the infinite-dimensional
		framework and it is satisfied in some standard situations.} For example, if  $\tau_{\chi,0}(s)\neq \mathbb{I}_{(0,1]}(s)$ (being  $\mathbb{I}_{(0,1]}$ the indicator function on the set
	$(0, 1]$)  and $\tau_{\chi,0}$ is continuously differentiable on $(0, 1)$, or if $\tau_{\chi,0}(s)=\delta_1(s)$ (where $\delta_1$ stands for the Dirac mass at $1$), then $M_{\chi,0}>0$ for any
	kernel $K$ satisfying  (C4). For further details on  condition (C6), we refer to  \cite{ferraty2007nonparametric}.
	
	The following theorem  {provides} the asymptotic distribution of the nonparametric regression estimator proposed in (\ref{C_est}).
	\begin{theorem}
		\label{eq:theorem2}
		Under assumptions \textnormal{(C1)--(C6)}, it can be proved that {, as $n\to\infty,$}
		$$\sqrt{n{F}_{\chi}(h)}\dfrac{\ell(\chi)M_{\chi,1}}{\sqrt{\sigma^2_1(\chi)M_{\chi,2}}}[\hat{m}_{h}(\chi)-m(\chi)-\bm{B}_h]\xrightarrow[]{\mathcal{L}}N(0,1),$$
		where $\xrightarrow[]{\mathcal{L}}$ denotes convergence in distribution, with
		$$\bm{B}_h= \varphi_{\chi}'(0)\dfrac{M_{\chi,0}}{M_{\chi,1}}h.$$
	\end{theorem}
	
	A simpler version of the previous theorem can be derived just by considering the following additional assumption:
	
	\begin{enumerate}[{(C}7)]
		\item  $\lim_{n\to\infty}h\sqrt{n{F}_{\chi}(h)}=0$.
	\end{enumerate}
	
	Condition (C7) allows the cancellation of the bias term  {$\bm{B}_h$} in Theorem  \ref{eq:theorem2}. Under this assumption (and also if \textnormal{(C1)--(C6)} hold), the pointwise asymptotic  {Gaussian} distribution for the functional nonparametric regression
	estimate is given in Corollary \ref{eq:cor1}.
	
	\begin{corollary}
		\label{eq:cor1}
		Under assumptions \textnormal{(C1)--(C7)}, it can be proved that {, as $n\to\infty,$}
		$$\sqrt{n{F}_{\chi}(h)}\dfrac{\ell(\chi)M_{\chi,1}}{\sqrt{\sigma^2_1(\chi)M_{\chi,2}}}[\hat{m}_{h}(\chi)-m(\chi)]\xrightarrow[]{\mathcal{L}} N(0,1).$$
	\end{corollary}
	
	Notice that in order to use the asymptotic distribution of $\hat{m}_{h}$ in practice, for inferential purposes, it is necessary to
	{estimate the functions $\ell$ and {$\sigma_1^2$}. Moreover, both constants $M_{\chi,1}$ and $M_{\chi,2}$ have to be computed. Assuming a uniform kernel function, $K(u)=\mathbb{I}_{[0,1]}(u)$, it is obvious that $M_{\chi,1} = M_{\chi,2} = 1$.}
	
	\begin{corollary}
		\label{eq:cor2}
		Under assumptions \textnormal{(C1)--(C7)},  {considering} $K(u)=\mathbb{I}_{[0,1]}(u)$ and if   {$\hat{\sigma}_1^2$ an $\hat{\ell}$ are consistent estimators of ${\sigma}_1^2$ and $\ell$, respectively}, it can be proved that
		$$\sqrt{n\hat{F}_{\chi}(h)}\dfrac{ {\hat{\ell}(\chi)}}{\sqrt{\hat{\sigma}^2_1(\chi)}}[\hat{m}_{h}(\chi)-m(\chi)]\xrightarrow[]{\mathcal{L}} N(0,1) \text{ as } n\to\infty,$$
		where 
		$$\hat{F}_{\chi}(h)=\dfrac{1}{n}\sum_{i=1}^n {\mathbb{I}_{\{{\norm{\mathcal{X}_i-\chi}\le h}\}},}$$
		{being $\mathbb{I}_{\{A \}}=1$ if $A$ is true and 0 otherwise.}
	\end{corollary}
	
	\begin{remark}  {An immediate consequence of Corollary \ref{eq:cor2} is {the possibility of cons\-truc\-ting (asymptotic) confidence intervals} for the circular regression function.  Fixing a confidence level $1-\alpha$, {with} $\alpha\in(0,1)$, it can be easily obtained that an {asymptotic} confidence interval for  $m$ at $\chi\in E$ is: 
			{$$\left(\hat{m}_{h}(\chi)\mp z_{\alpha/2}\frac{1}{\sqrt{n\hat{F}_{\chi}(h)}}\frac{\sqrt{\hat{\sigma}^2_1(\chi)}}{ {\hat{\ell}(\chi)}}\right),$$}
			where $z_{\alpha /2}$ denotes the $(1-\alpha)$-quantile of the standard normal distribution {and $\hat\sigma_1^2(\chi)$ and $\hat\ell(\chi)$ are (consistent) estimators of $\sigma_1^2(\chi)$ and $\ell(\chi)$, respectively.}}
		
		Among the possible (consistent) conditional variance estimators for {$\sigma_1^2$}, we suggest {to use the following method based on adapting to the functional-circular framework}  the approach studied in \cite{mario_var}. {In the present context,} taking into account that
		\begin{eqnarray*}
			\label{var_cond}
			\sigma^2_1(\chi)={\rm Var}[\sin(\varepsilon)\mid\mathcal{X}=\chi]
			=\mathbb{E}[\sin^2(\varepsilon)\mid\mathcal{X}=\chi]-\{\mathbb{E}[\sin(\varepsilon)\mid\mathcal{X}=\chi]\}^2,
		\end{eqnarray*}	
		and under the assumption that $\mathbb{E}[\sin(\varepsilon)\mid\mathcal{X}=\chi]=0$, only $\mathbb{E}[\sin^2(\varepsilon)\mid\mathcal{X}=\chi]$ has to be estimated. For this, in a first step, considering model \eqref{model}, the residuals from a nonparametric fit (using, for instance, the Nadaraya--Watson estimator with a pilot bandwidth) are obtained. Then, in a second step, the estimator of the conditional variance function {$\sigma_1^2$} is defined as the nonparametric estimator (employing again the Nadaraya--Watson estimator with another bandwidth) of the regression function using the squared sine of the residuals as the response variables.
		
		{Regarding function $\ell$, taking into account that $\ell(\chi)=[m^2_1(\chi)+m_2^2(\chi)]^{1/2}$, a consistent estimator for this function is straightforwardly obtained replacing in this equation $m_j$ by the Nadaraya--Watson estimators $\hat{m}_{j,h}$, $j=1,2$, defined in \eqref{C_estNadaraya--Watson}. }

	\end{remark}

	\section{Simulation study}\label{sec:sim}
	{In this section, the regression estimator  proposed in (\ref{C_est}) is analyzed nu\-me\-ri\-ca\-lly by simulation,} considering different scenarios and model (\ref{model}). For each scenario, 500 samples of size $n$ ($n=50, 100, 200$ and $400$) are generated   {using} the following regression  {functions}:
	{
		\begin{eqnarray*}
			\text{r1: } m(\mathcal{X})&=&{\rm atan2}(0.5+\textstyle\int^\mathbb{I}_0\mathcal{X}(t)dt,\textstyle\int^\mathbb{I}_0\mathcal{X}(t)dt),\\
			\text{r2: } m(\mathcal{X})&=&{\rm acos}(-0.3\textstyle\int^\mathbb{I}_0\mathcal{X}(t)dt)+3/2{\rm acos}(0.4\textstyle\int^\mathbb{I}_0\mathcal{X}(t)dt),
		\end{eqnarray*}
	}
	where the random curves are simulated from 
	\begin{equation}
		\mathcal{X}(t)=30U(1-t)t^{1+U}, \quad t\in[0,1],\label{Xt}\end{equation}
	being $U$ a uniformly distributed variable on $[0,1]$ {, and t}he circular errors $\varepsilon$  {in (\ref{model})} are drawn from a von Mises distribution $vM(0,\kappa)$,  {with} different  concentrations ($\kappa=5, 10$ and $15$). 
	
		\begin{figure*}[htb]
	\centering
	\includegraphics[width=1\textwidth]{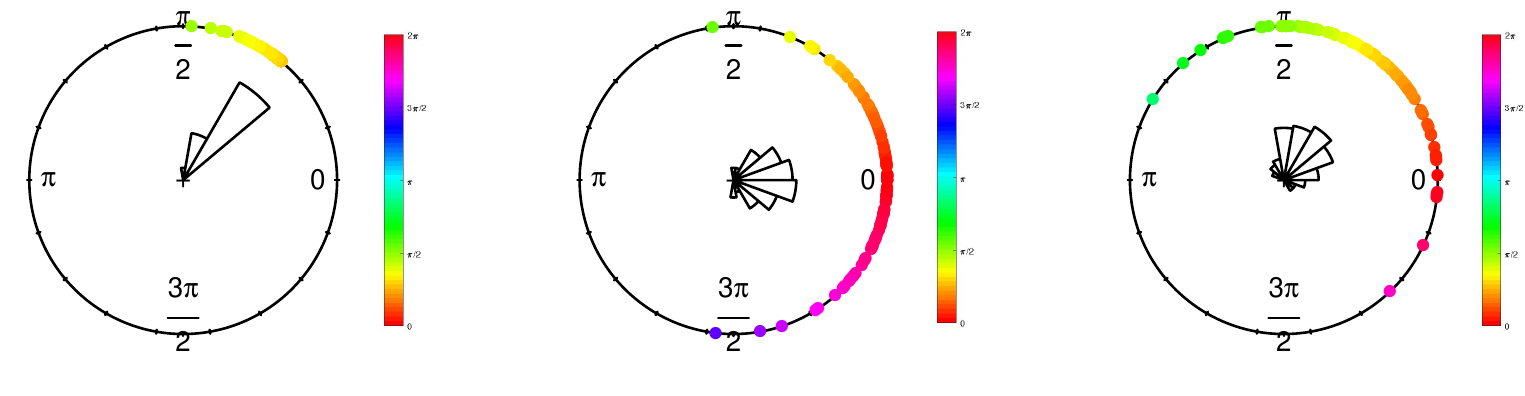}\vspace{0.6cm}
	\includegraphics[width=1\textwidth]{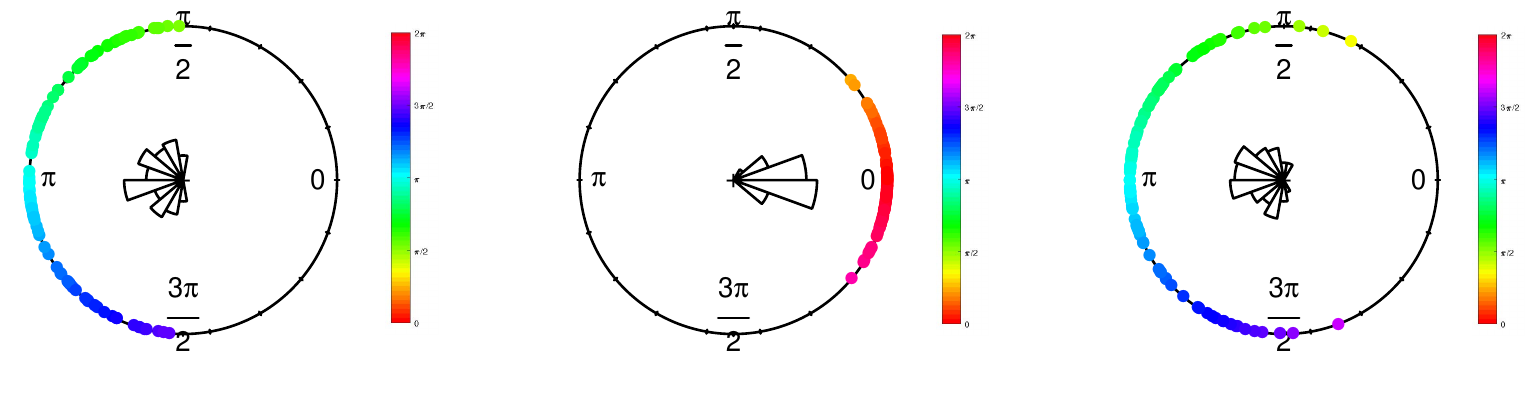}
	\caption{Illustration of model generation (model  {with regression function r1}: top row; model  {with regression function r2}: bottom row). Left: regression function evaluated at a random sample of size $n=100$ of curves simulated from (\ref{Xt}). Center: independent errors from a von Mises distribution with zero mean and concentration $\kappa=5$ (for model with regression function r1) and $\kappa=15$ (for model with regression function r2). Right: random responses  obtained by adding the two previous plots.}
	\label{figure:process}
\end{figure*}
	
		\begin{table}[h!]
		\begin{center}
			\caption{Average  CASE given in (\ref{CASE}), over 500 replicates, using a cross-validation bandwidth, ${h}_{{\rm CV}}$,  and an optimal bandwidth, ${h}_{{\rm CASE}}$, as a benchmark. Data are generated from  models with regression functions r1 (left) and r2 (right) and von Mises errors with different concentration parameters ($\kappa=5, 10, 15$).}
			\begin{tabular}{cccccccccccc}
				&  &  & &  {r1} &  &&&&& {r2}& \\ 
				\cline{4-6}\cline{10-12}
				$\kappa$&$n$&&$h_{{\rm CV}}$&&$h_{{\rm CASE}}$&&&&$h_{{\rm CV}}$&&$h_{{\rm CASE}}$\\
				\hline
				5 & 50 &  & 0.0096 &  &0.0037 &&&&0.0122&& 0.0087 \\ 
				& 100 &  & 0.0046 &  & 0.0021 &&&&0.0096&& 0.0069\\ 
				& 200 &  & 0.0011 &  & 0.0005 &&&&0.0052&& 0.0047\\ 
				& 400 &  &  0.0004&  &  0.0002&&&&0.0030&& 0.0023\\ 
				\hline10 & 50 &  & 0.0020 &  & 0.0018&&&&0.0072&& 0.0058\\ 
				& 100 &  & 0.0010 &  & 0.0007 &&&&0.0062&&0.0046\\ 
				& 200 &  &0.0005  &  &  0.0003 &&&&0.0033&&0.0028 \\ 
				& 400 &  & 0.0004 &  &  9e-05 &&&&0.0026&&0.0021\\ 
				\hline15 & 50 && 0.0019    &  & 0.0017 &&&&0.0056&& 0.0048 \\ 
				& 100 &  &0.0007  &  &  0.0005 &&&&0.0056&&0.0039\\ 
				& 200 &  & 0.0004 &  & 0.0003 &&&&0.0021&&0.0019 \\ 
				& 400 &  & 0.0002 &  &  8e-05
				&&&&0.0015&&0.0013\\ 
				\hline
			\end{tabular}
			\label{table}
		\end{center}
	\end{table}
	
	{As an illustration,} Figure \ref{figure:process} shows two realizations of simulated data of size $n=100$ (model  {with regression function r1} in top row and model  {with regression function r2} in bottom row). Left plots display the circular regression  {functions} evaluated in the curves generated from (\ref{Xt}). Central panels present the random errors drawn from a von Mises distribution with zero mean direction and concentration $\kappa=5$ (for model  {with regression function r1}) and $\kappa=15$ (for model  {with regression function r2}).  It can be seen that the errors in the top row, corresponding to $\kappa=5$,  present more variability than the ones generated with $\kappa=15$. Right panels show the values of the response variables, obtained {adding}   {the mean values in the left panels} and  {the} circular errors  {in the central panels}.

	{For each simulated sample, the Nadaraya--Watson-type estimator  {(\ref{C_est})} is computed using the {kernel $K(u)=1-u^2$, $u \in (0,1)$,} and taking the $L^2$ metric to calculate the distance between curves. Regarding the bandwidth $h$, it is selected by cross-validation, choosing the bandwidth parameter ${h}_{{\rm CV}}$ that minimizes 
		the function:
		\begin{equation}\label{cv}{\rm CV}(h)=\sum_{i=1}^n \left\{  1-\cos\left[\Theta_i-\hat{m}_{h}^{(i)}(\mathcal{X}_i)\right]\right\},\end{equation}
		where $\hat{m}_{h}^{(i)}$  {denotes} the Nadaraya--Watson-type estimator computed using all observations except for $(\mathcal{X}_i,\Theta_i)$.}

	{To evaluate the performance of the Nadaraya--Wat\-son-type estimator $\hat{m}_h$, the circular average squared error (CASE), defined by \cite{kim2017multivariate}:
		\begin{eqnarray}\label{CASE}
			{\rm CASE}[\hat{m}_{h}(\chi)]=\frac{1}{n}\sum_{i=1}^n \left\{1- \cos\left[m(\mathcal{X}_i) - \hat{m}_{h}(\mathcal{X}_i)\right]\right\},\end{eqnarray}
		is computed.}
	Table \ref{table} shows the average (over  500 replicates) of the CASE
	when $h$ is selected by cross-validation. For comparative purposes, {in each scenario, the average of the minimum value of the CASE, computed using the bandwidth $h_{{\rm CASE}}$ that minimizes (\ref{CASE}), is also included in Table \ref{table}. Note that the bandwidth $h_{{\rm CASE}}$ can not be computed in a practical situation, because the regression function is unknown.}  {The corresponding CASE errors are included in this table as a benchmark}. 
	In the different scenarios, it can be seen that the average errors decrease when the sample size {increases}. In addition, as expected, results are  better when the error concentration gets larger.

	The appropriate performance of the nonparametric circular regression estimator \eqref{C_est} is also observed in Figure \ref{figure:est}. In this {plot}, {the same scenarios as those considered to obtain Figure \ref{figure:process} are used. In the left,} the theoretical regression functions  {r1 and r2} {are show and, in the right,} the corresponding estimates {using} a random sample {generated from these models are presented.} The re\-pre\-sen\-ta\-tions in the top row correspond to the data simulated {using the model with regression function}  {r1} and those in the bottom row {with regression function}  {r2}.
	
	\begin{figure*}[htb]
		\centering
		\includegraphics[width=0.75\textwidth]{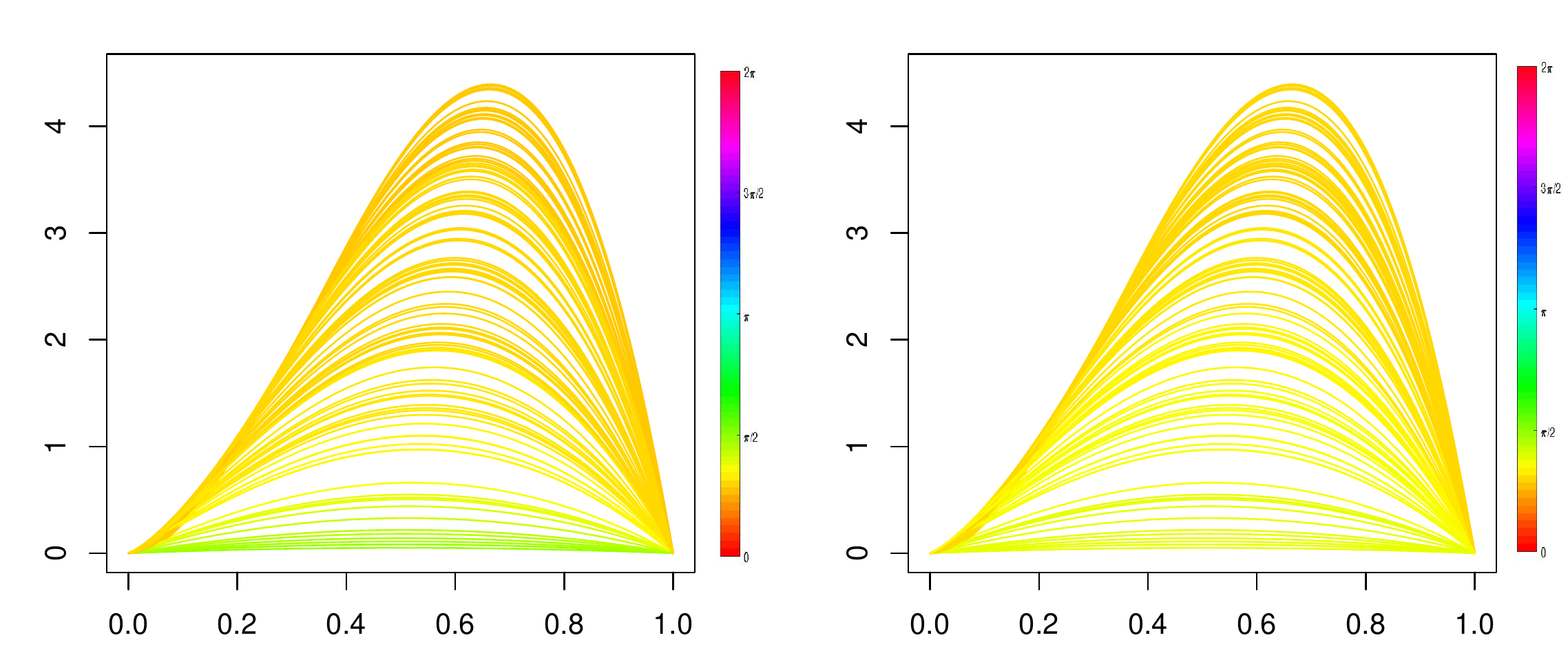}
		\includegraphics[width=0.75\textwidth]{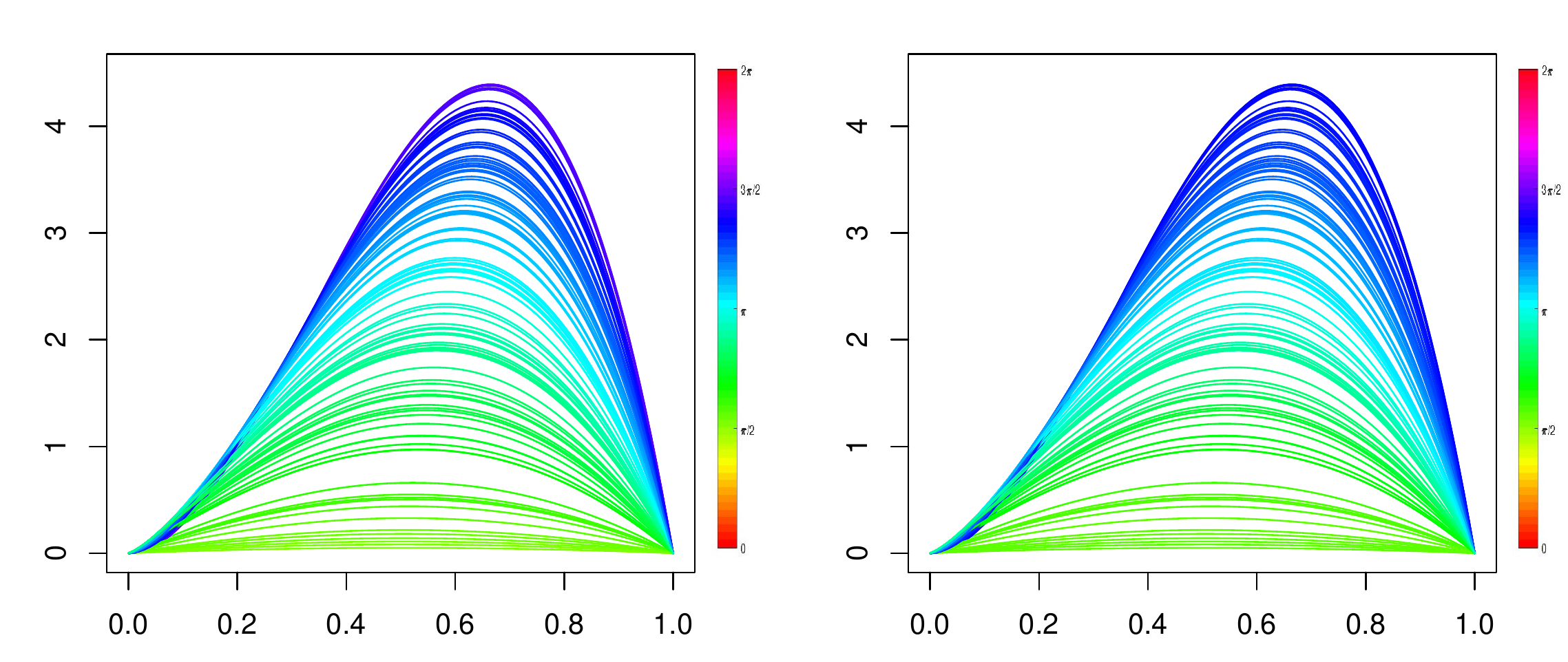}
		\caption{Theoretical regression function (left), jointly with the corresponding estimates (right), using the specific scenarios considered in Figure \ref{figure:process}, for the model with regression function r1 (top row) and r2 (bottom row).}
		\label{figure:est}
	\end{figure*}
	
	{To complete the study, we repeated the simulations using a \textit{k}-nearest neighbor (\textit{k}NN) version of our estimator \citep[see][for a study of the \textit{k}NN method in a regression model with scalar response a functional covariate]{Burba2009}. The results obtained were very similar to those achieved by the Nadaraya--Watson approach, with a slightly better performance of the Nadaraya--Watson-type estimator. For reasons of space, only the results with r1 using the \textit{k}NN-type estimator are presented (Table \ref{table2}). For comparison, we also include in this table the part of Table \ref{table} relative to r1.}
	
	\begin{table}[htb]
		\begin{center}
			\caption{{Average  CASE given in Equation (19), over 500 replicates, using the \textit{k}NN-type estimator (\textit{k}NN). Data are generated from the  model with regression function r1 and von Mises errors with different concentration parameters ($\kappa=5, 10, 15$). For comparison, results obtained with the Nadaraya--Watson-type estimator in this scenario, presented in Table \ref{table}, are also included.}}
			\begin{tabular}{cccccccc}
				$\kappa$&$n$&&$h_{{\rm CV}}$&&$h_{{\rm CASE}}$&&{\textit{k}NN}\\
				\hline
				5 & 50 &  & 0.0096 &  &0.0037 && 0.0472\\ 
				& 100 &  & 0.0046 &  & 0.0021 &&0.0152\\ 
				& 200 &  & 0.0011 &  & 0.0005 &&0.0070\\ 
				& 400 &  &  0.0004&  &  0.0002&&0.0036\\ 
				\hline10 & 50 &  & 0.0020 &  & 0.0018&&0.0216\\ 
				& 100 &  & 0.0010 &  & 0.0007 &&0.0073\\ 
				& 200 &  &0.0005  &  &  0.0003 && 0.0035\\ 
				& 400 &  & 0.0004 &  &  9e-05 &&0.0018\\ 
				\hline15 & 50 && 0.0019    &  & 0.0017 &&0.0143 \\ 
				& 100 &  &0.0007  &  &  0.0005 &&0.0048\\ 
				& 200 &  & 0.0004 &  & 0.0003 && 0.0024\\ 
				& 400 &  & 0.0002 &  &  8e-05
				&&0.0013\\ 
				\hline
			\end{tabular}
			\label{table2}
		\end{center}
	\end{table}

	\section{Real data illustration}
	\label{sec:example}
	{There is a quite extensive literature on the computation of daily temperature averages or profiles, which are useful to summarize temperature patterns \citep[see][]{Huld2006, Ma2013, Bernhardt2018}. More recently, there has been some interest in investigating the identification of daily temperature extremes \citep[see][]{Glynis2021} or the association of diurnal temperature range with mortality burden \citep[see][]{Cai2021}. All these approaches are possible given the high frequency registries of temperatures that are available nowadays, claiming for a proper statistical analysis using functional data methods (see \citealp{Ruiz2012}, for a spatial autoregressive functional approach, and \citealp{Aguilera2017}, for spatial functional prediction). These works, taylored for a spatial analysis, do not consider the calendar time as a response.}
	
	{Daily temperature curves (obtained by data records every 10 minutes, so each curve has 144 points) have been collected by Meteogalicia\footnote{Meteogalicia: https://www.meteogalicia.gal} in Santiago de Compostela (NW-Spain) from February 15, 2002 until December 31, 2019. The northern region of Spain presents an oceanic climate, characterized by cool winters and warm summers, with mild temperatures in {spring and fall} seasons. For the initial period from February 15, 2002 to June 28, 2005{, using the available sample $(\mathcal{X}_i,\Theta_i)$, $i=1, \ldots, n$, where $\mathcal{X}_i$ denotes the whole daily temperature curve and $\Theta_i$ the corresponding day on a year basis, for the day $i$ in that period} (see {Figure~\ref{figure:rd}} in the Introduction), the regression estimator proposed in (\ref{C_est}) is {computed}. The fitted model is obtained {using the same kernel {and metric} as in the simulations} and considering the {cross-validation} bandwidth selector {obtained minimizing (\ref{cv})}, which gives $h_{{\rm CV}}=8.8139$. As it has been explained {in the Introduction}, relevant changes in temperature patterns could be revealed by comparing 
		{the predicted days with the fitted model for some temperature curves, in a different period, with the observed days corresponding to those temperature curves.} {Note that the period selected for estimating the model contains enough data to proceed with our nonparametric method, and it is far enough in time to mitigate temporal correlation from the prediction experiment carried out in what follows.}
	}
	
	As a specific example, consider the period from May 21 to May 27, 2019.
	For this {spring} week, the {7} observed curves can be seen in Figure~\ref{figure:estrd} ({left}), with colors indicating the observed days. In Figure~\ref{figure:estrd} (right), same {7} observed curves are now colored according to the predicted days where these curves should be observed with the model fitted from the 2002-2005 data in Figure~\ref{figure:rd} (left). According to the model for this early period, the observed curves in the considered week in May 2019 correspond to days in June, July and August. So, this week in May 2019 would be perceived as warmer with respect to the reference period.
	
	\begin{figure*}[htb]
		\centering
		\includegraphics[width=1\textwidth]{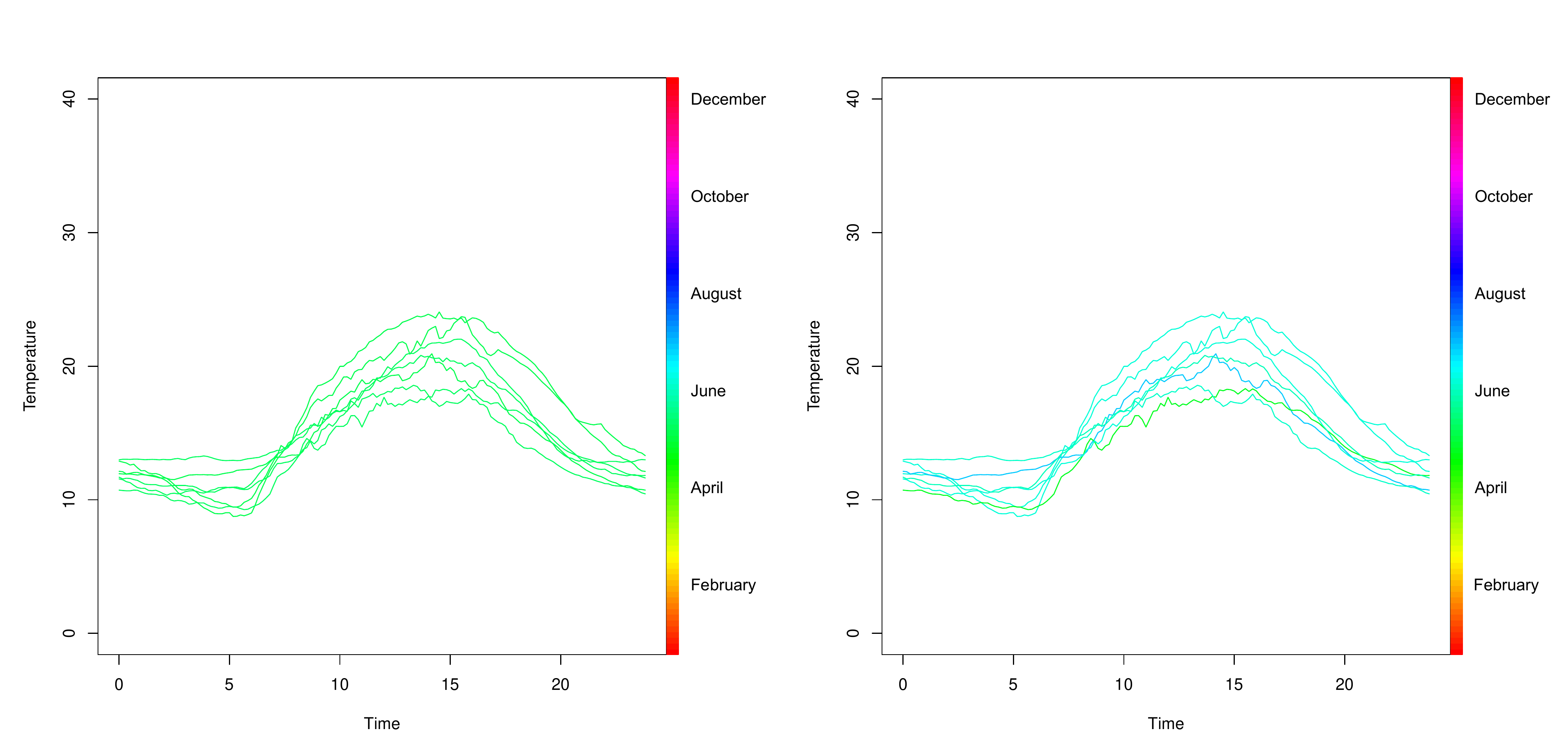}
		\caption{Daily temperature curves in Santiago de Compostela (Spain) from May 21, 2019 to May 27, 2019 (left). Colors indicate the observed day. Temperature curves colored by predicted days using the fitted model (right).}
		\label{figure:estrd}
	\end{figure*}

	{For a more general analysis, consider the temperature curves registered along 2019 and take a division by months. {We denote} the observed sample {by} $(\mathcal{X}^b_j,\Theta^b_j)$, for $j=1,\ldots,n_b$ {and $b=1,\ldots,12$}, being $n_b$ the number of days of month $b$. Similarly to the measurement error considered for evaluating the performance of the estimator in (\ref{CASE}), a circular average prediction error (CAPE) is defined for each month as:}
	\begin{equation}
		\mbox{CAPE}_b=\frac{1}{n_b}\sum_{j=1}^{n_b}\left\{1-\cos\left[{\Theta^b_j}-\hat m_h({\mathcal{X}^b_j})\right]\right\},
		\label{CAPE}
	\end{equation}
	{where $\hat m_h(\mathcal{X}^b_j)$, for $j=1, \ldots, n_b$ and $b=1,\ldots, 12$, is the predicted day in the year 2019 with the fitted model using {the} 2002-2005 data. Note that this prediction error is given in an easily interpretable scale since a perfect match between observed and predicted days for a sample of temperature curves yields to a null value; if predictions are located 3-months later/before than observations, then the prediction error takes value 1 (so a month time lag gives one third); if predictions and observations are 6-months apart, then the prediction error takes value 2. Boxplots for monthly prediction  errors for 2019 are depicted in Figure~\ref{figure:CAPE}.}  
	
	\begin{figure*}[htb]
		\centering
		\includegraphics[width=1\textwidth]{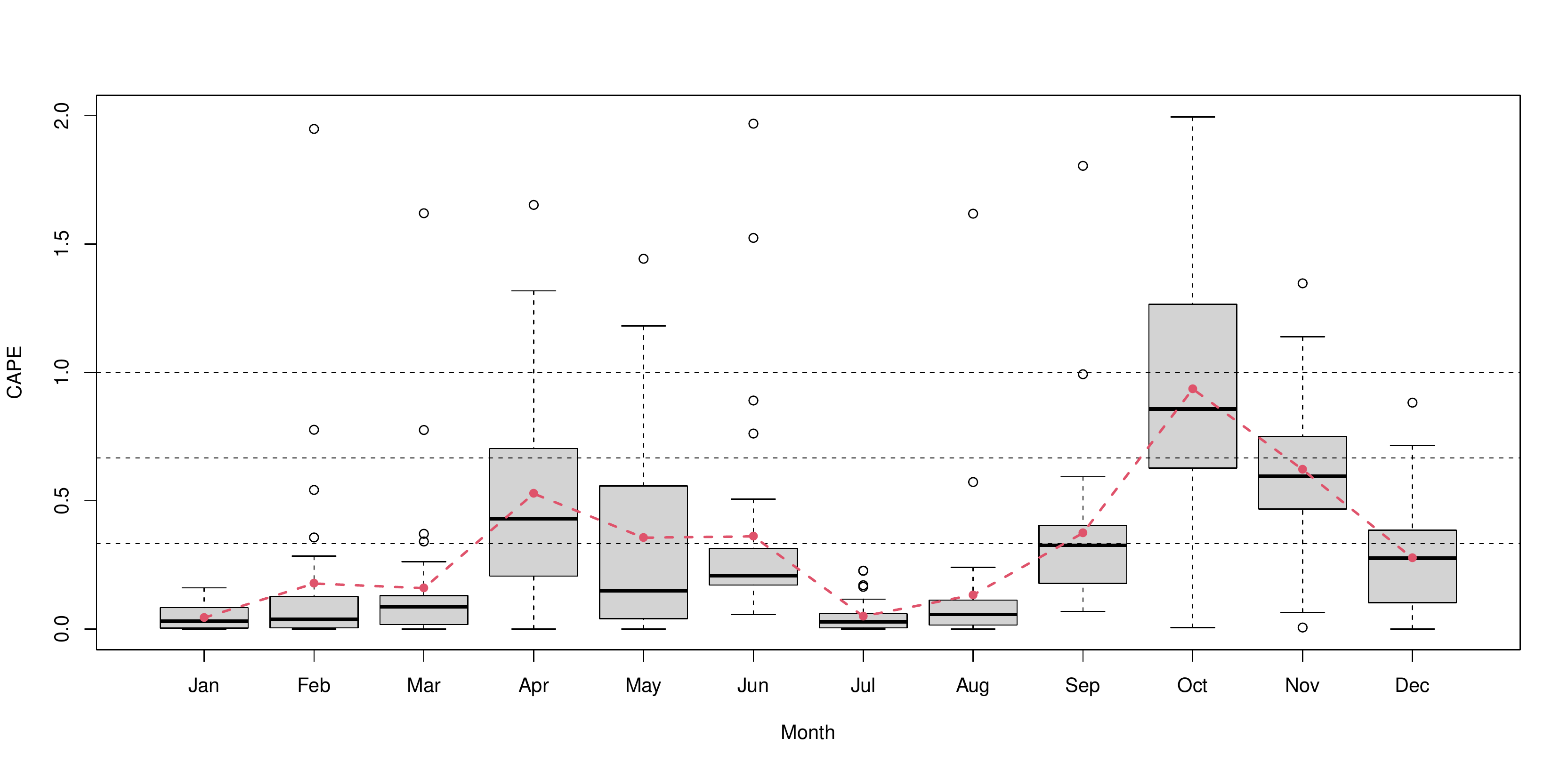}
		\caption{Boxplots for the circular prediction errors, by months, for 2019. Red dots: circular average prediction error for each month.}
		\label{figure:CAPE}
	\end{figure*}

	\begin{figure*}[h!]
	\centering
	\includegraphics[width=1\textwidth]{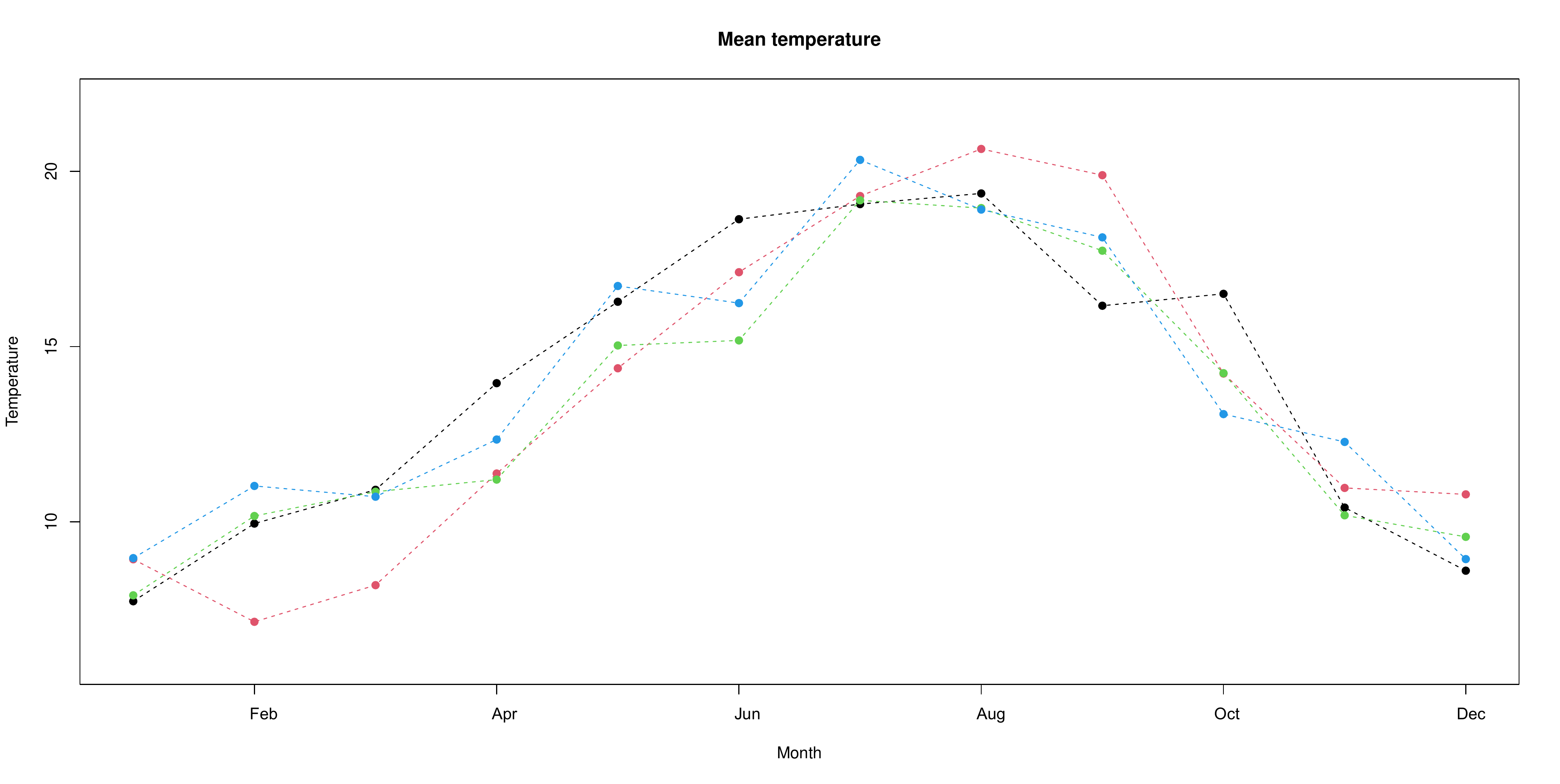}
	\caption{Mean temperature in Santiago de Compostela (Spain) in 2017 (black), 2018 (red), 2019 (green) and 2020 (blue) by months. }
	\label{figure:mean_temp}
\end{figure*}
	
	{Note that there are seven months (April, May, June, September, October, November and December) that present prediction errors larger than 1/3, indicating that the mismatch between the observed day for a certain curve and the predicted day (with the model fitted for the early period) is larger than one month. 
		{On the other hand, it is observed in Figure \ref{figure:CAPE} that January/February  and July/August in 2019 behave, more or less, as usual Januaries/Februaries  or Julies/Augusts in the years considered in the training sample, perhaps with only a small ``shift" in days for the temperature curves of these months. The experts from the meteorological service consulted by us indicated that this is the expected behavior for these months in a context of global warming. In the Atlantic climate (the one corresponding to Galicia, where the data were recorded), one of the manifestations of climate change is the disappearance of seasons, specifically, spring and fall. January/February and July/August are still the coolest and warmest months, respectively, in the region considered in this research (see Figure \ref{figure:mean_temp} for the years 2017-2020) and they would be the  months that would distinguish the ``periods" of the year. } To verify the possible slight displacement in days corresponding to the temperature curves for January and July pointed out by the experts, we compared the days observed and the days predicted by our model considering the temperature curves for January and July 2019, as it was performed in Figure \ref{figure:estrd} (see Figures \ref{figure:jan_tem} and \ref{figure:jul_tem}). In these two figures, the effect indicated by the experts seems to be observed. A similar behavior was observed for the temperature curves of February  and August.

	\begin{figure*}[htb]
		\centering
		\includegraphics[width=1\textwidth]{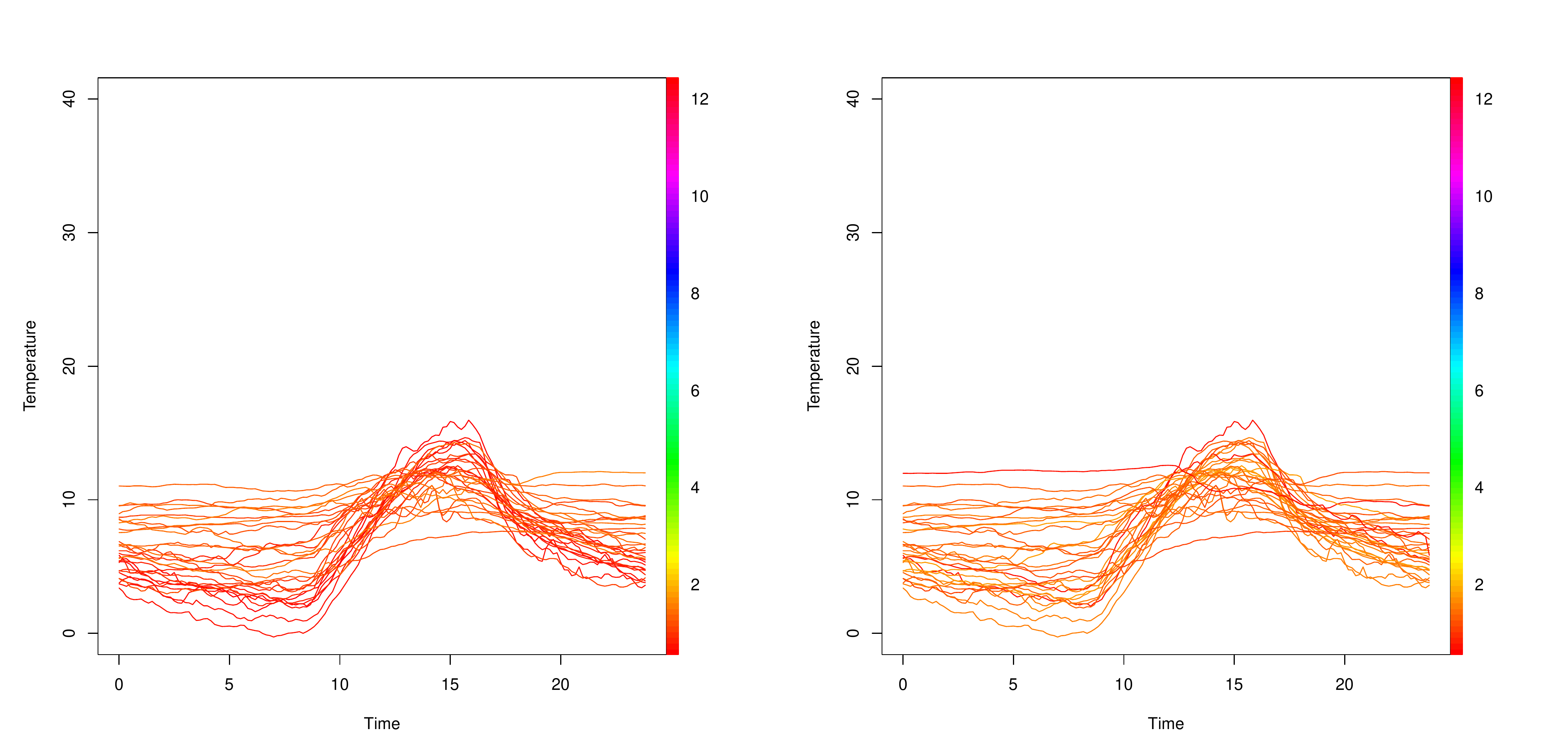}
		\caption{Daily temperature curves in Santiago de Compostela (Spain) from January 1, 2019 to January 31, 2019 (left). Colors indicate the observed day. Temperature curves colored by predicted days using the fitted model (right).}
		\label{figure:jan_tem}
	\end{figure*}
	
	\begin{figure*}[h!]
		\centering
		\includegraphics[width=1\textwidth]{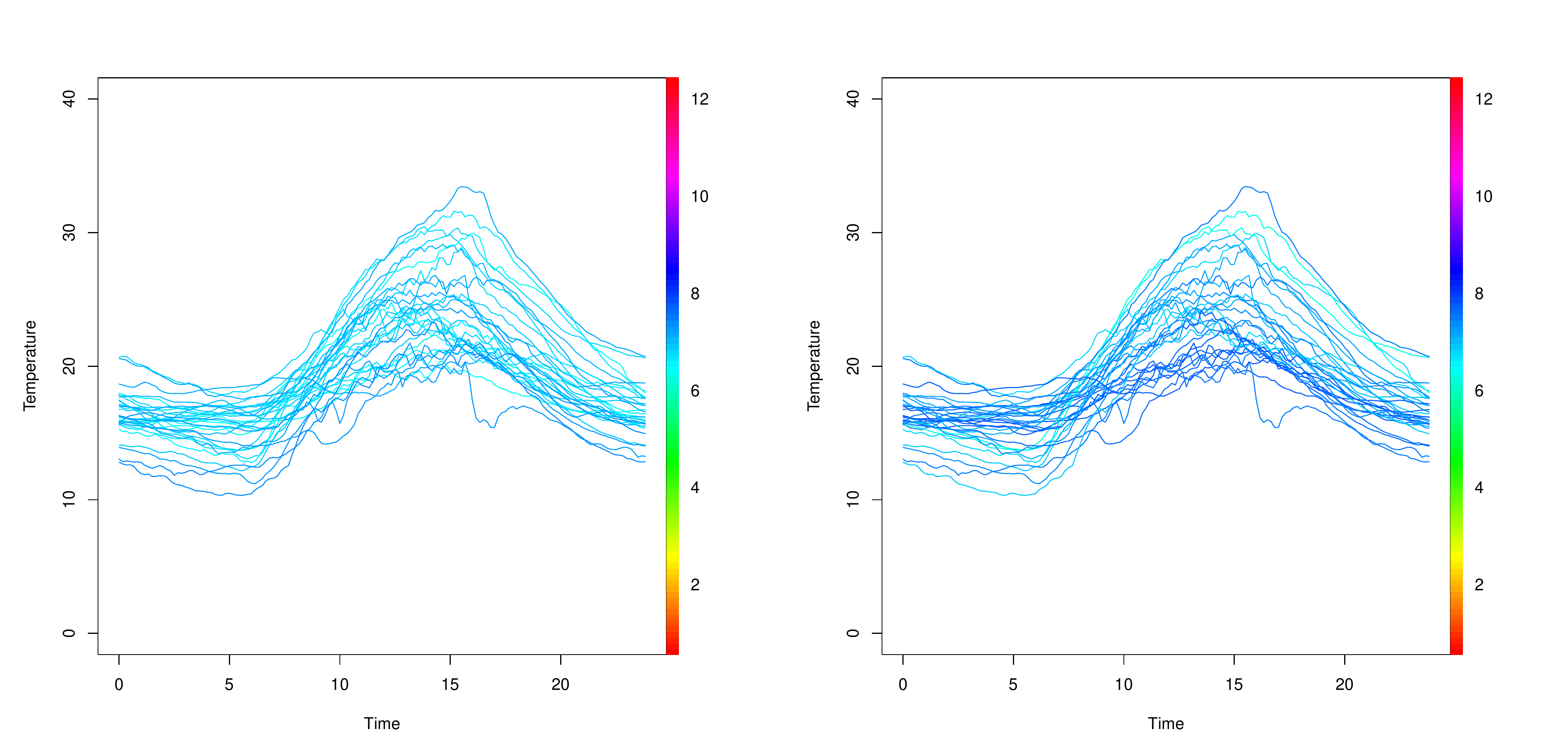}
		\caption{Daily temperature curves in Santiago de Compostela (Spain) from July 1, 2019 to July 31, 2019 (left). Colors indicate the observed day. Temperature curves colored by predicted days using the fitted model (right).}
		\label{figure:jul_tem}
	\end{figure*}

		{Moreover, to analyze more in deep the direction of the displacements observed in Figure \ref{figure:CAPE}, circular boxplots \citep{buttarazzi2018boxplot} of the differences between the observed days and those predicted by our nonparametric estimator are computed. Figure \ref{fig:boxplots_circ} shows these circular bloxplots, for all the months, that is, the circular boxplots of ${\Theta^b_j}-\hat m_h({\mathcal{X}^b_j})$, for $j=1,\ldots, n_b$ and $b=1,\ldots,12$. Note the large variability in predictions for April, May, June and October, as well as the small variability but with a clear shift (prediction errors are not centered at zero) for September (warmer than expected according to the fitted model), November and December (cooler than expected). Additionally, to rule out that year 2019 was a singular year, we replicated this experiment for years 2017, 2018 and 2020, showing a very similar pattern in all cases (see Figure \ref{figure:boxplot_all}).}

		\begin{figure*}[htb]
			\centering
			\includegraphics[scale=0.19]{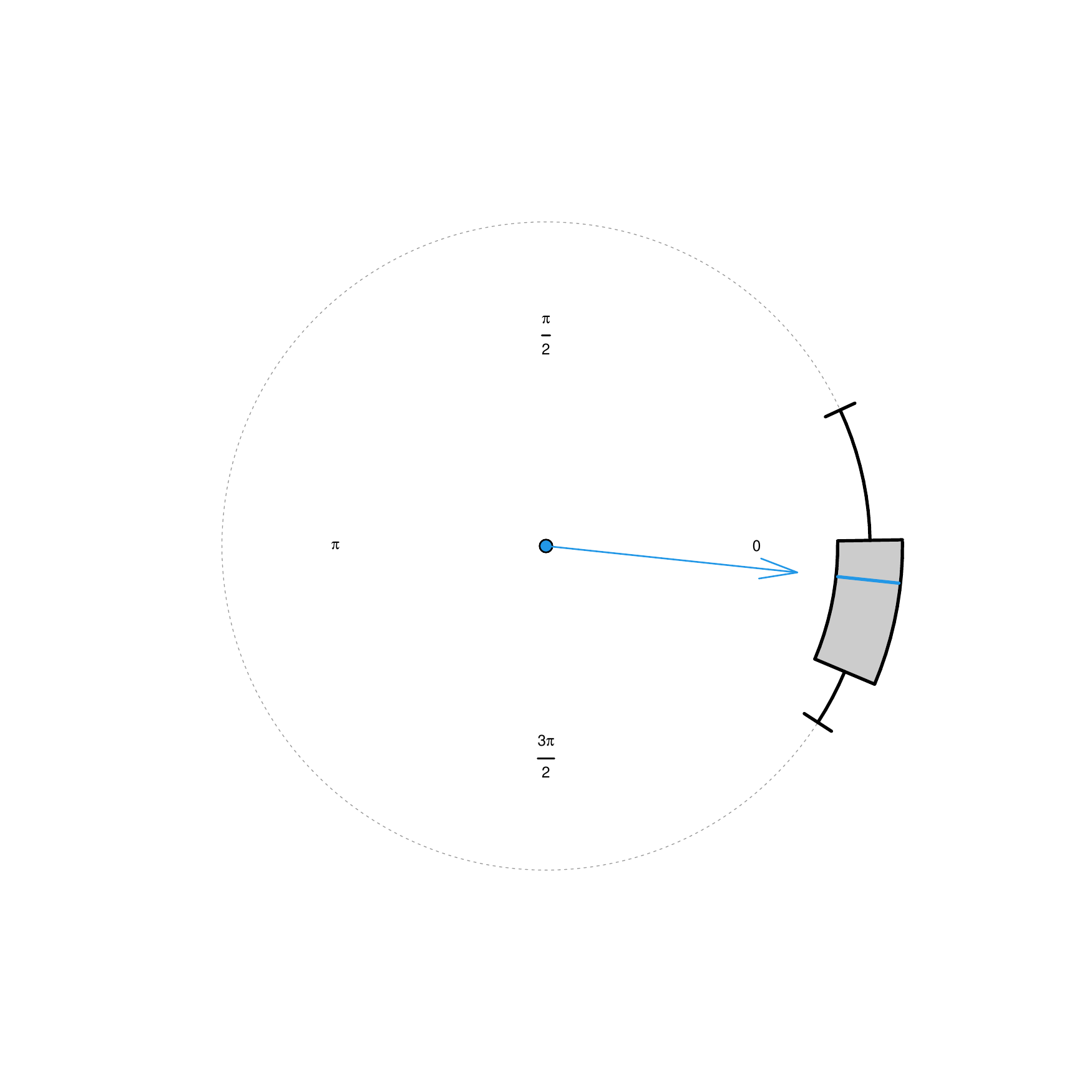}
			\includegraphics[scale=0.19]{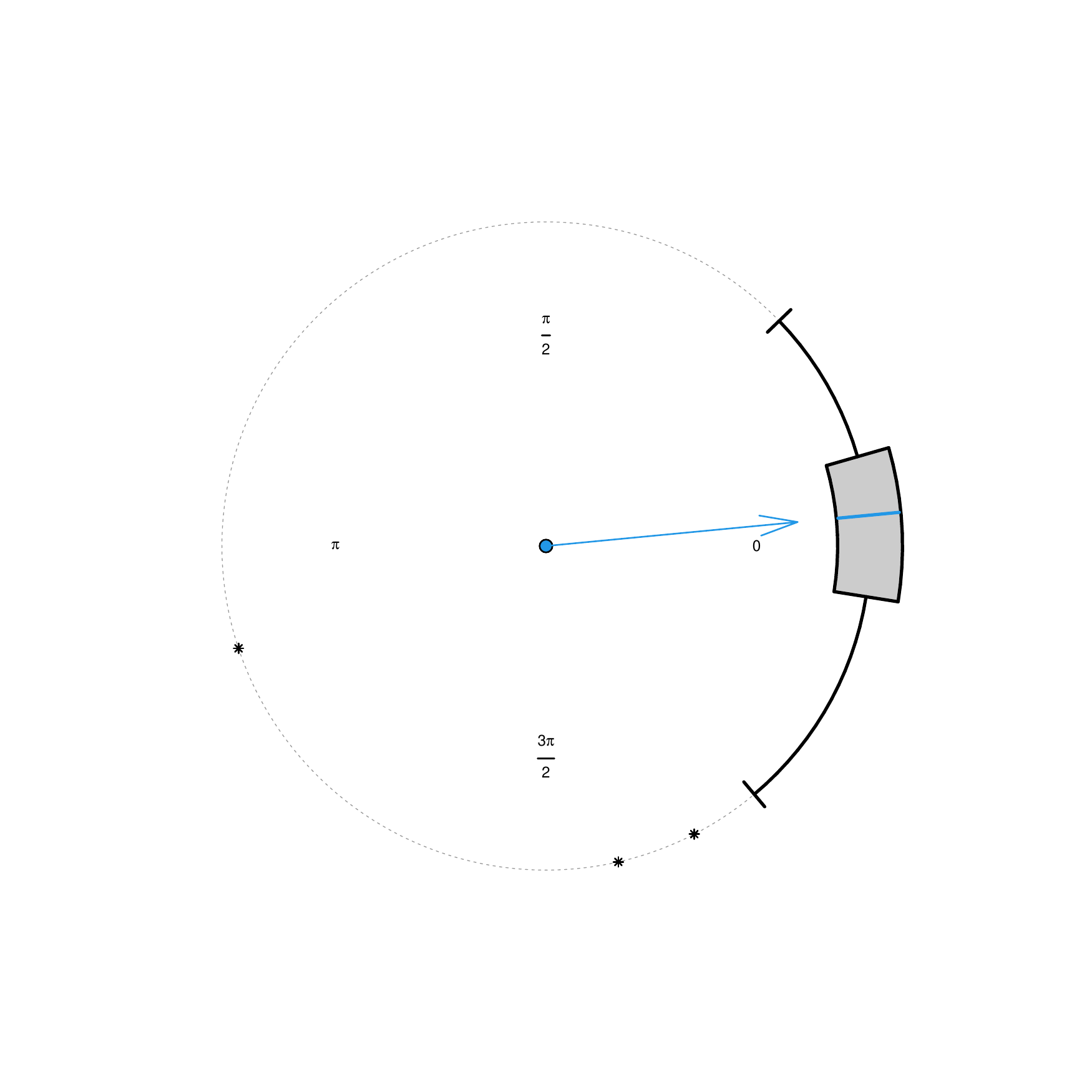}
			\includegraphics[scale=0.19]{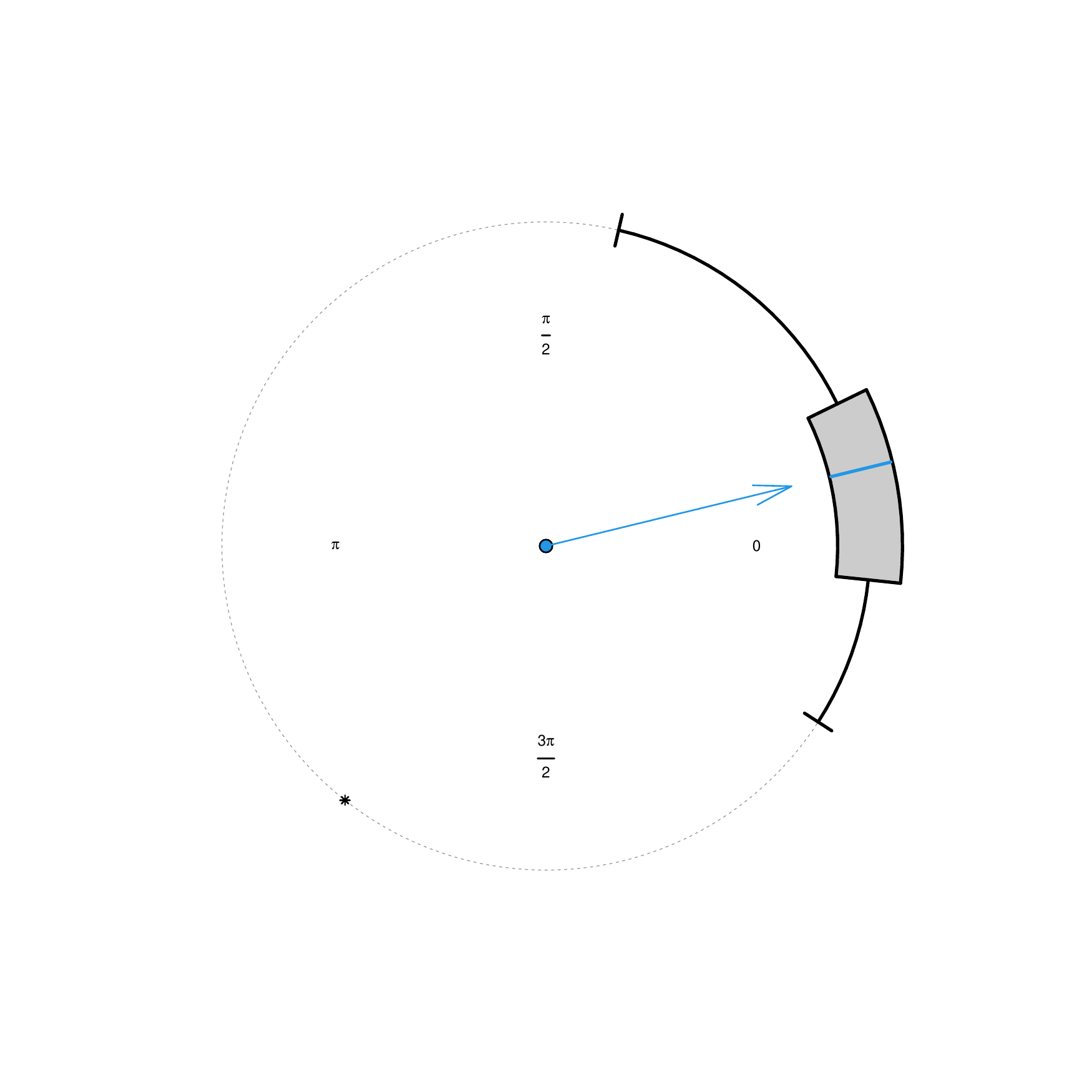}
			\includegraphics[scale=0.19]{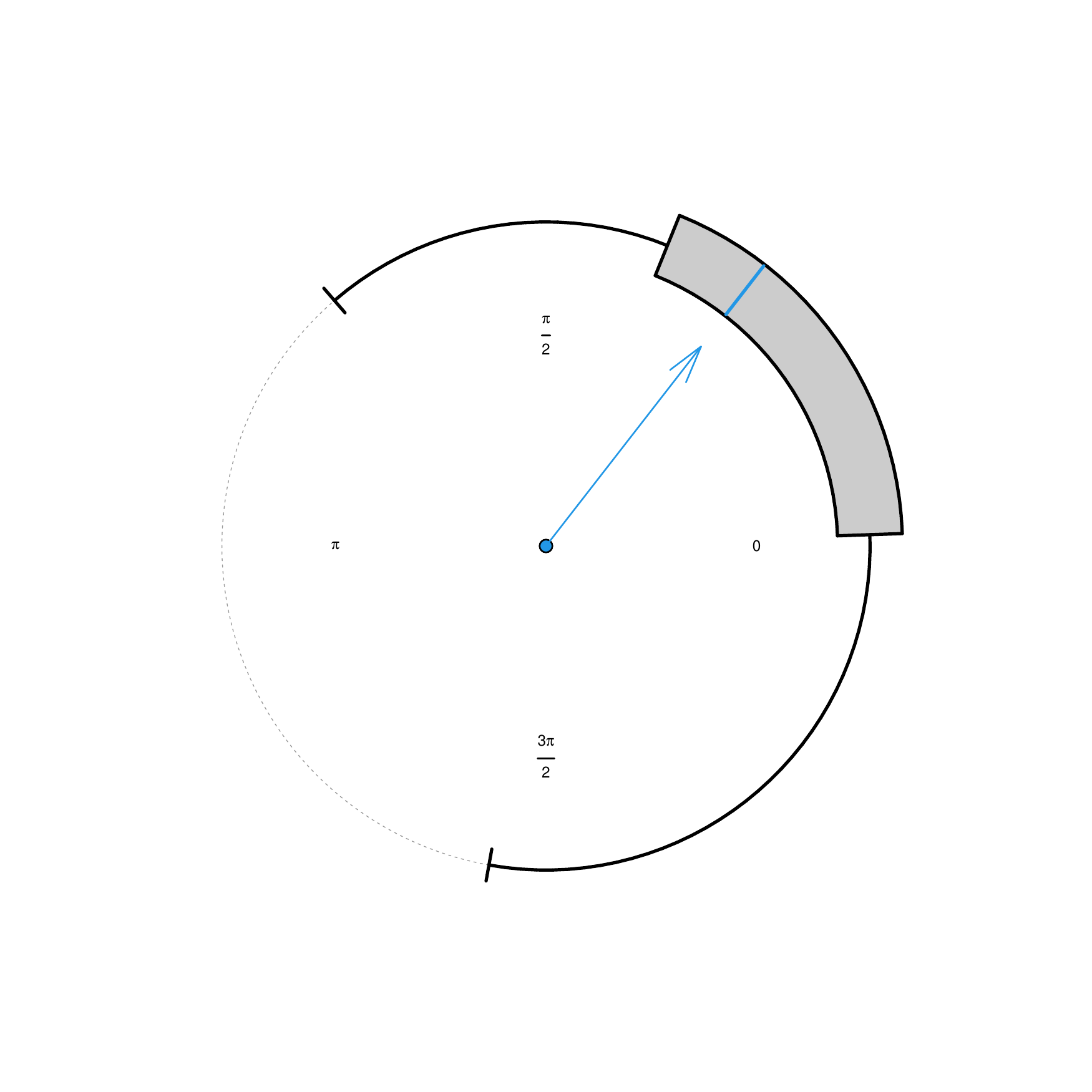}
			\includegraphics[scale=0.19]{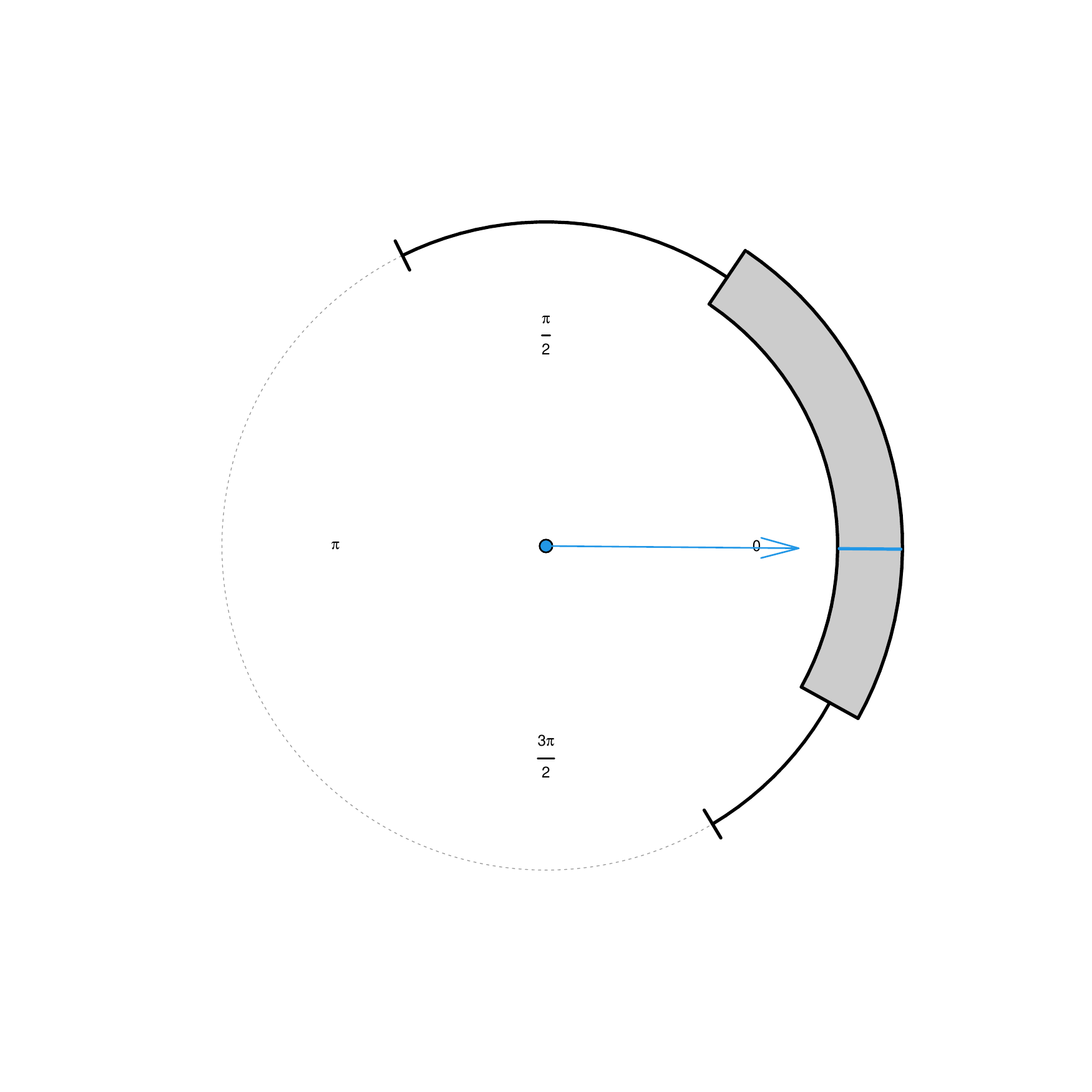}
			\includegraphics[scale=0.19]{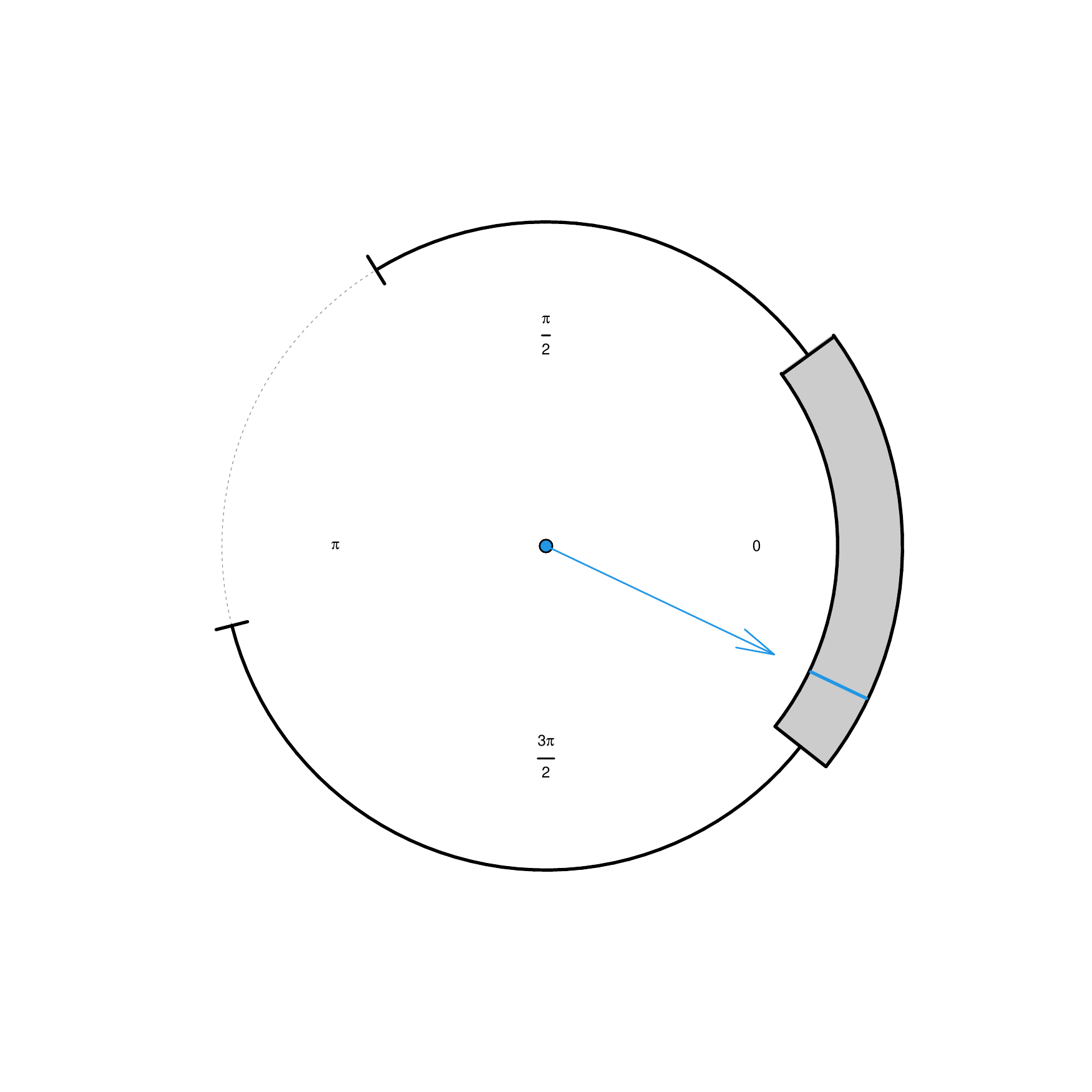}
			\includegraphics[scale=0.19]{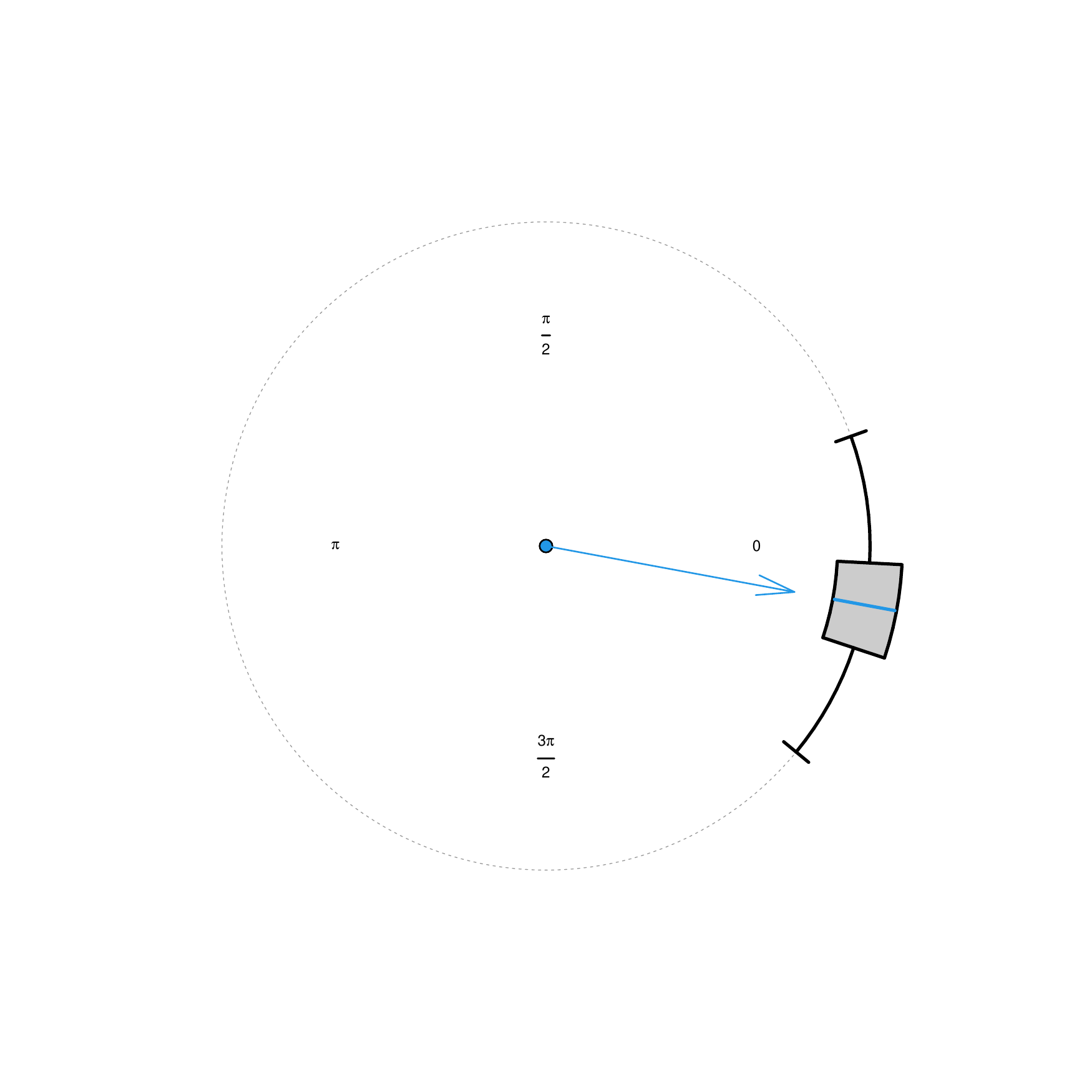}
			\includegraphics[scale=0.19]{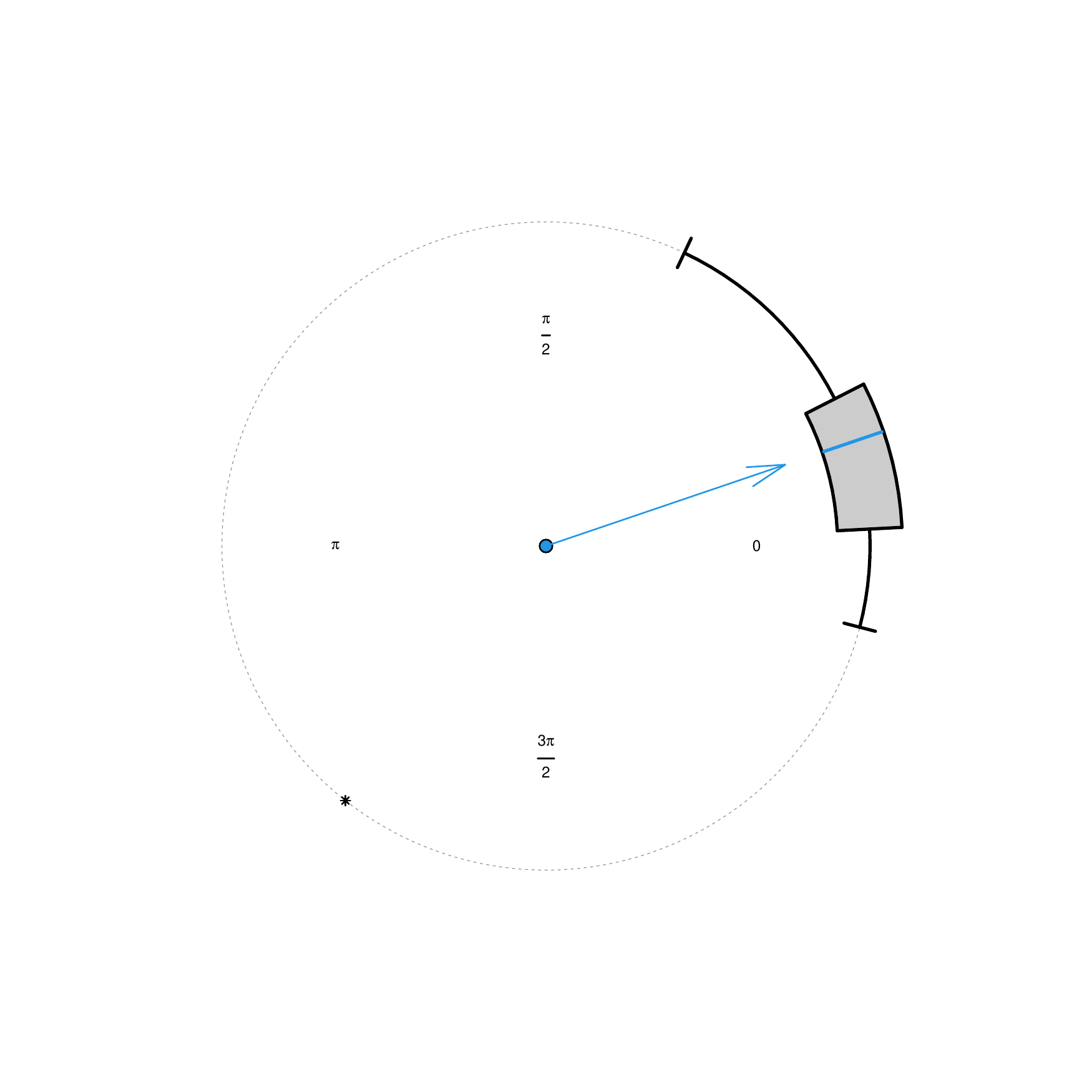}
			\includegraphics[scale=0.19]{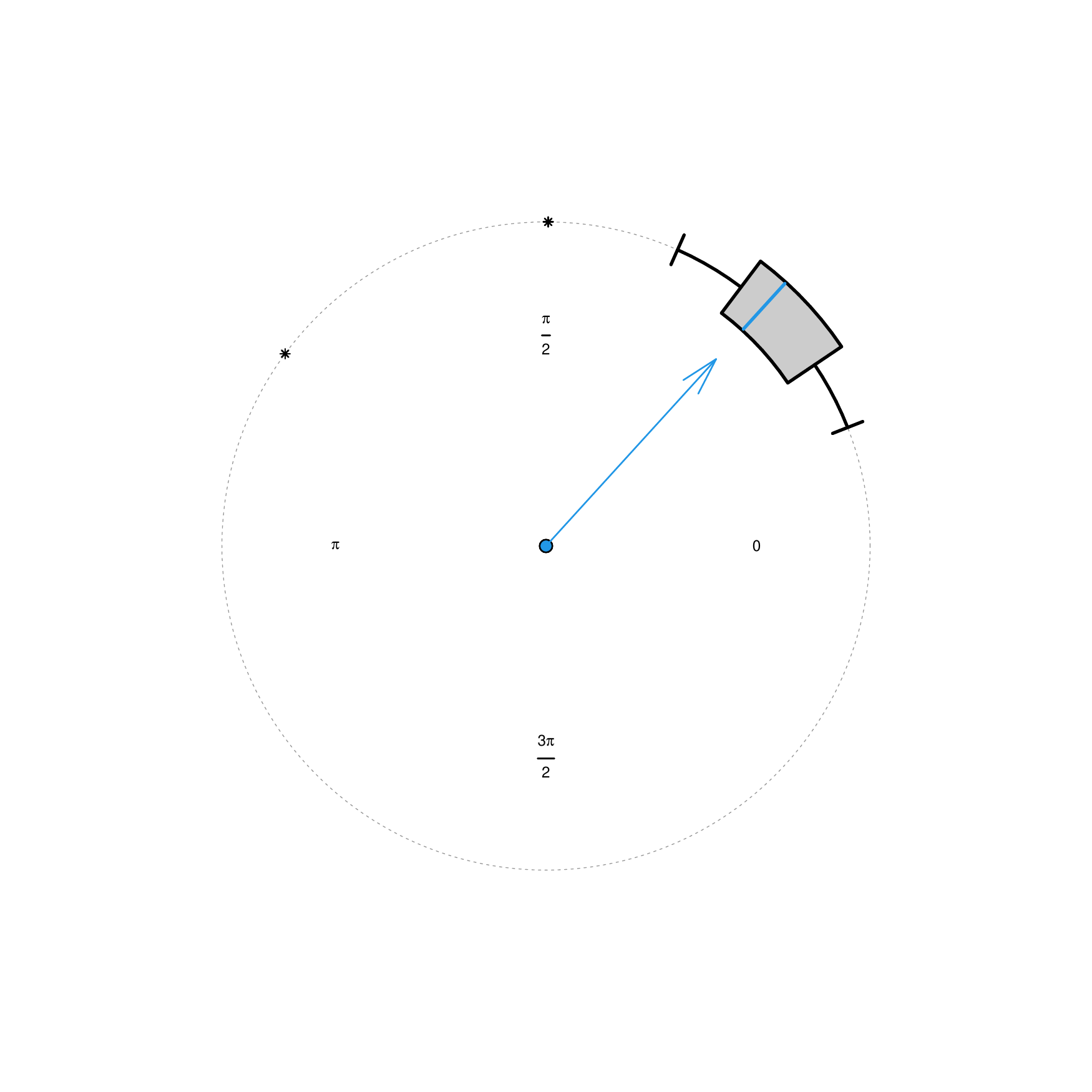}
			\includegraphics[scale=0.19]{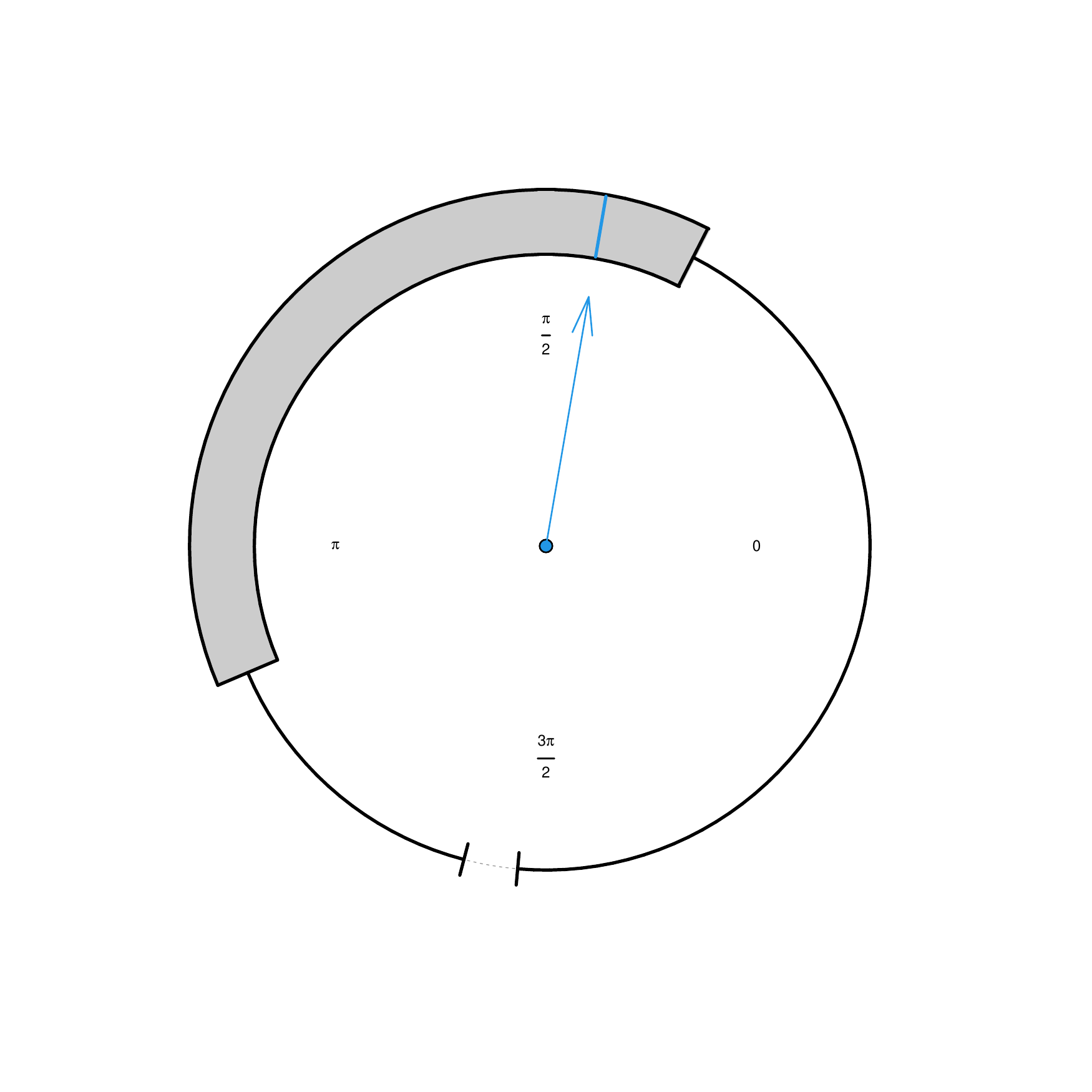}
			\includegraphics[scale=0.19]{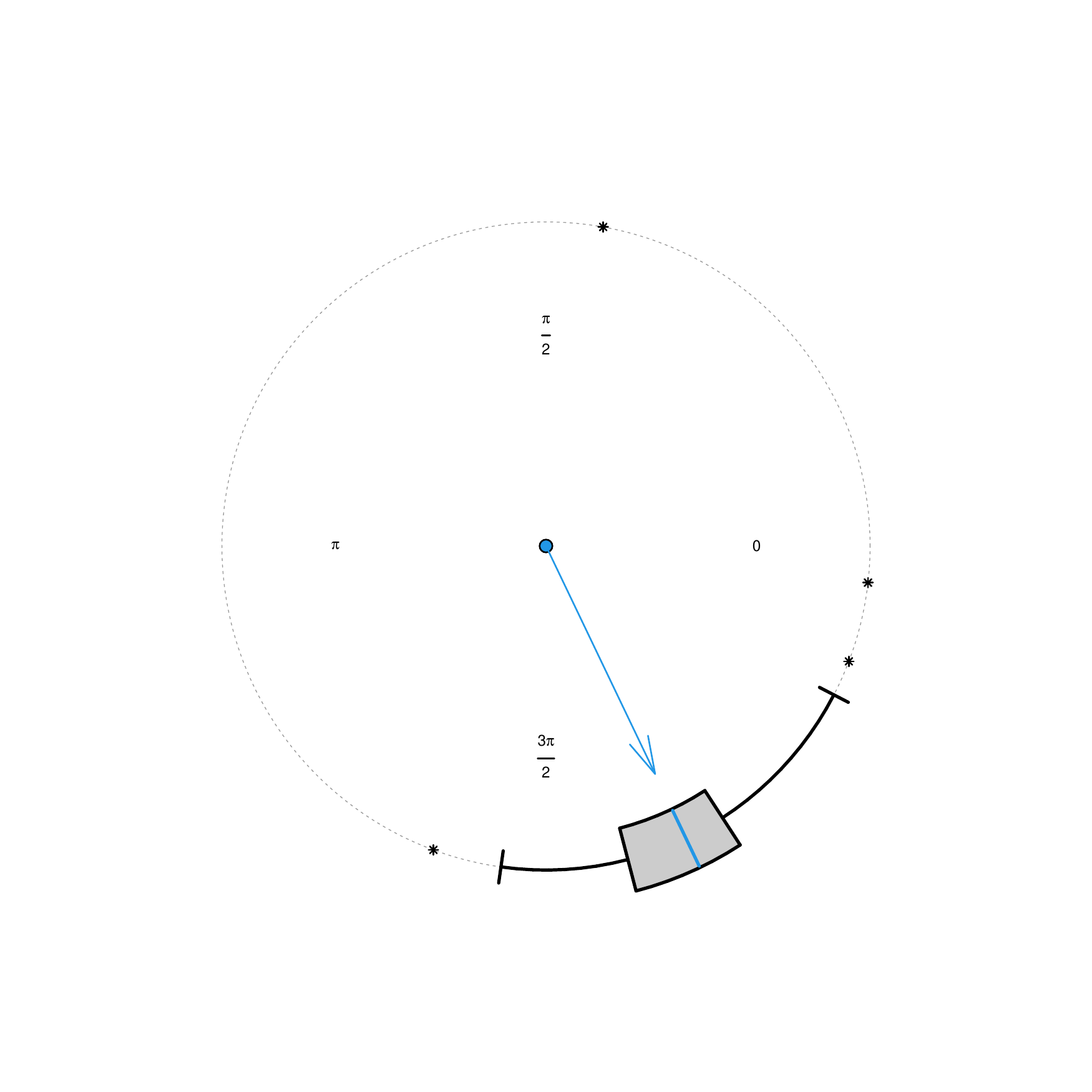}
			\includegraphics[scale=0.19]{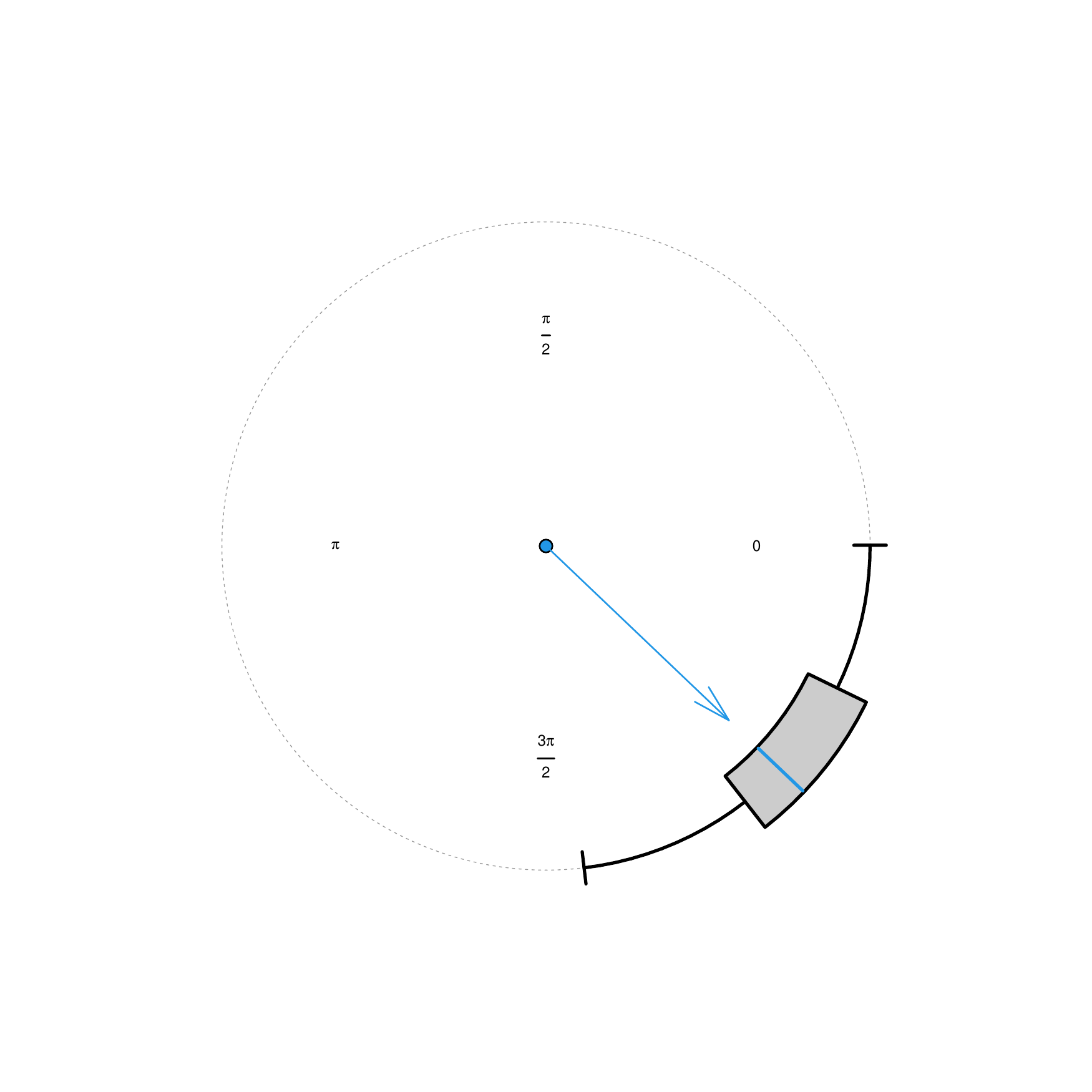}
			\caption{Circular boxplots for the circular prediction errors, by months, for 2019. The months are in order from left to right and from top to bottom.}
			\label{fig:boxplots_circ}
		\end{figure*}
	
	\begin{figure*}[h!]
		\centering
		\includegraphics[width=1\textwidth]{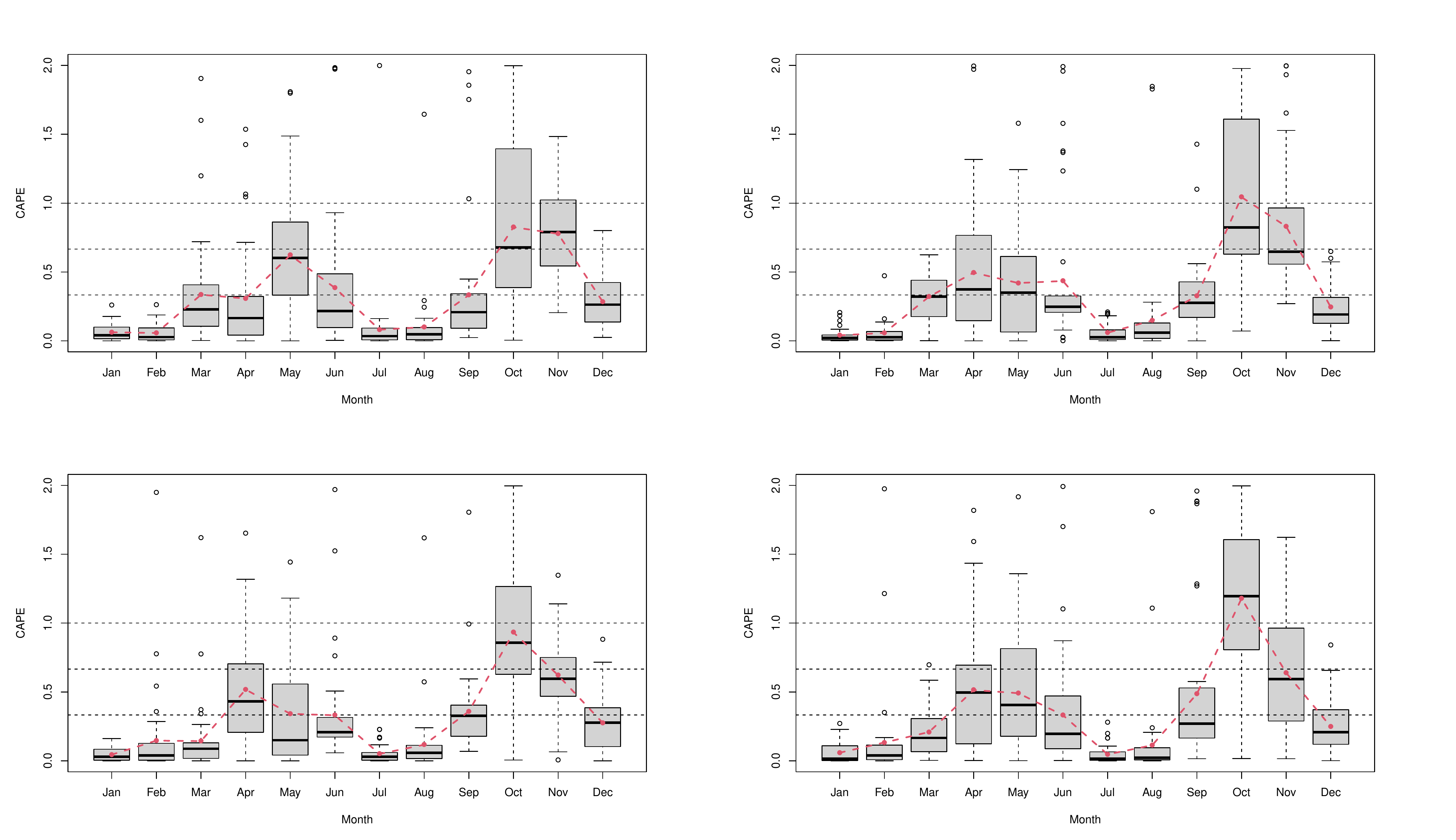}
		\caption{Boxplots for the circular prediction errors, by months, for 2017 (top left), 2018 (top right), 2019 (bottom left) and 2020 (bottom right). Red dots: circular average prediction error for each month. }
		\label{figure:boxplot_all}
	\end{figure*}
		
		{We also repeated this study using the \textit{k}NN-type estimator compared in the simulations. The results and conclusions derived are practically the same to those obtained with the Nadaraya--Watson-type estimator. Figure~\ref{fig:boxplots} shows the same information as that presented in Figure \ref{figure:CAPE}, but using the \textit{k}NN-type estimator.}
		
		{As a final comment, it should be highlighted the potential use of the proposed method in order to explain climate change effects on agriculture. The IPCC 2022 \citep[][Chapter 5]{portner2022climate} indicates that temperature trend changes have modified the life cycle of crops (both shortening or prolonging them, depending on latitude). In some areas, this effects have lead farmers to change their agricultural practices. In addition, the IPCC 2022 \citep{portner2022climate} also notices that temperature variability directly affects both the characteristics of some harvest (e.g. acidity or colour of fruits) and the risk of pests and diseases. Understanding temperature variability along the year can be useful for adapting agricultural practices.}

		\begin{figure*}[htb]
			\centering
			\includegraphics[width=1\textwidth]{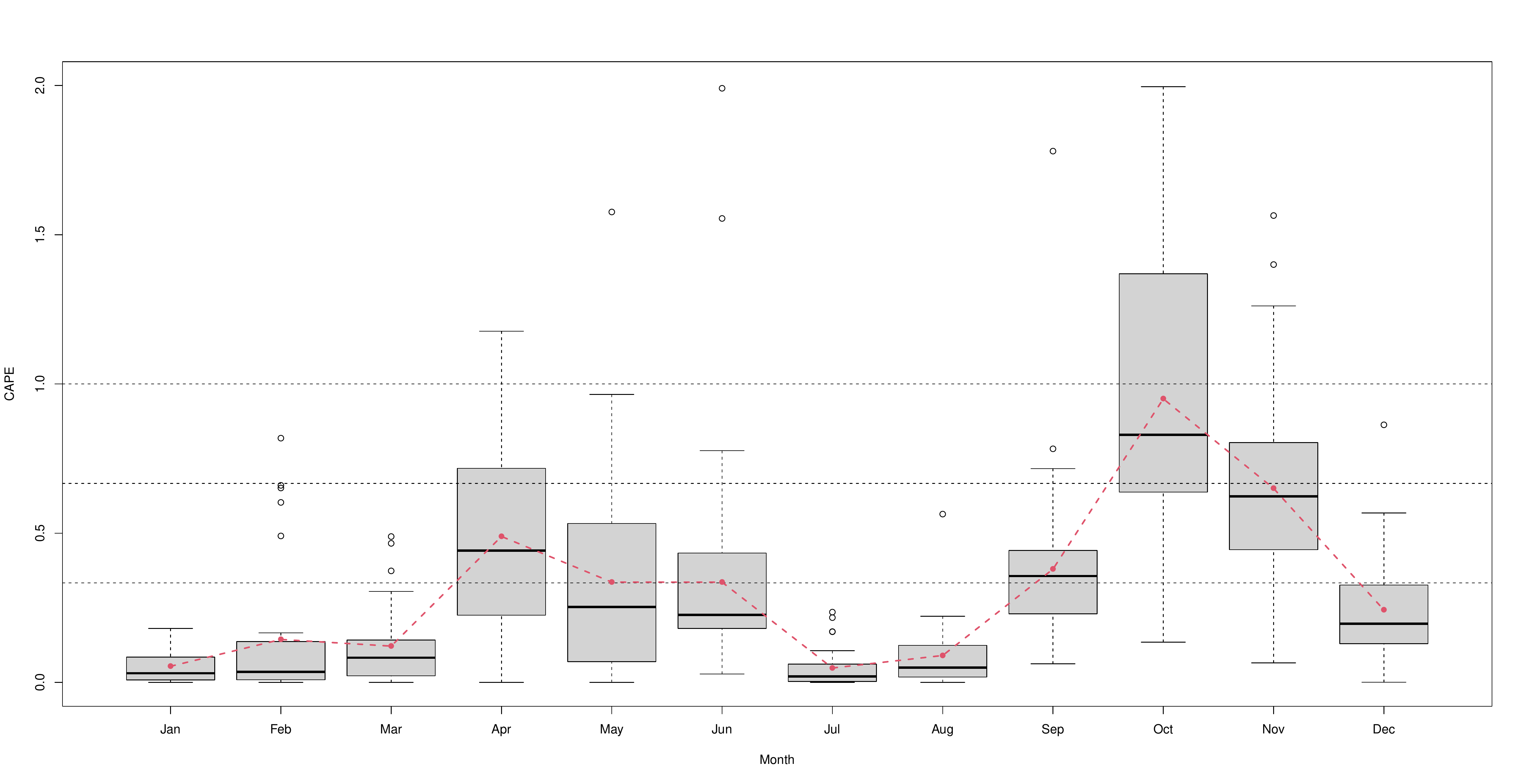}
			\caption{Boxplots for the circular prediction errors, by months, for 2019. Red dots: circular average prediction error for each month, using \textit{k}NN-type estimator.}
			\label{fig:boxplots}
		\end{figure*}

		\section{Discussion}\label{sec:dis}
		{Regression analysis of models involving}   {non-Euclidean data}  {represents a cha\-llen\-ging problem given the need of developing new statistical approaches for its study. In this paper, we focus on one}   {of these}  {non-standard settings. In par\-ti\-cu\-lar, we consider a regression model with a circular response and a functional covariate. A nonparametric estimator of Nadaraya--Watson type of the co\-rres\-pon\-ding regression function is proposed and analyzed from a theoretical and a practical point of view. Asymptotic bias and variance expressions, as well as asymptotic normality of the nonparametric approach are derived. The finite sample performance of the estimator is assessed by simulations and illustrated using a real data set.}   {From a practical perspective, our methodology allows to model the relation between a temperature curve (as a covariate) and a day (as a response), enabling to illustrate how temperature patterns change on a local scale.}
		
		{As in any kernel-based method {an} important tuning parameter, called the bandwidth or smoothing parameter, has to be selected by the user. In this paper, a cross-validation approach was proposed for this task and used in the numerical studies. Other possible selectors of plug-in type or based on bootstrap resampling methods could also be designed. While the use of global plug-in bandwidths in the infinite-dimensional context can be very intricate, given the difficulty of tackling with integrated versions of the AMSE, a bootstrap bandwidth selection method could be easily formulated. The idea would consist in approximating the CASE, given in \eqref{CASE} (or the mean of the CASE), by a bootstrap version of this error criterion using bootstrap samples obtained from bootstrap residuals. For this, two pilot bandwidths are needed, one to obtain the residuals and another to compute the bootstrap samples. This second pilot bandwidth can also be used to compute the Nadaraya--Watson{-type} estimator of $m$ employed instead of the theoretical regression function in the bootstrap expression of the CASE \eqref{CASE}. 
			Repeating this process many times, the bootstrap  bandwidth would be obtained by minimizing the averaged bootstrap CASE.
			Theoretical justification for this functional-circular bootstrapping procedure as well as the design of practical rules to obtain both pilot bandwidths are open problems out of the scope of the present paper, but of interest for a future research.}
		
		{In this paper, we have focused on a Nadaraya--Watson-type estimator for $m$ in model~(\ref{model}). However, a local linear-type estimator could be also defined. In that case,  local linear regression estimators for  models with a Euclidean response and a functional covariate are required. Moreover, using the 
			asymptotic theory on these estimators \citep{baillo2009local, ferraty2022scalar} and with similar arguments to those employed in the present paper for \eqref{C_est}, asymptotic results for the local linear-type regression estimator could also be derived.}
		
		{A possible extension of the model assumed in this paper would consist in including additional type of covariates. In a practical situation, this kind of complex models could help to obtain better predictions.  The most natural way to address this problem (following our kernel methodology) would be to use product kernels defined on possibly different spaces for the estimators of the sine and cosine components of the response.
			For instance, the ideas in \citet{garcia2013kernel} for cylindrical  density estimation, and the ones in \citet{racine2004nonparametric} or in  \citet{li_racine} for categorical data can be employed to define the corresponding weights in this context.   These type of product kernel estimators have been recently studied in nonparametric functional regression in \cite{Shang14} and in \cite{Selk2022}, for models with scalar response and different type of covariates (functional, real-valued or discrete-valued), and they could be directly employed in our circular-functional framework as regression estimators of the sine and cosine models. The important bandwidth selection problem in this context could be accomplished by a  cross-validation procedure or using a Bayesian approach \citep{Shang14}.
			We think that a deep study of these extended models is an interesting point that
			could be addressed as a further research.}
		
		{It should be noticed that {the temperature curves considered in our data illustration} exhibit temporal dependence. The nonparametric approach for estimating the regression function is not affected by data dependencies, but such an issue should be accounted for when obtaining {the asymptotic properties and a more appropriate bandwidth parameter. {C}onsidering that the proposed estimator is computed as the atan2 of nonparametric regression estimators for sine and cosine models, results about nonparametric regression for scalar responses and functional covariate with temporal dependence could be employed {to address this issue} \citep{attouch2010asymptotic,attouch2013robust,ferraty2006nonparametric,masry2005nonparametric,ling2019uniform,kurisu2022nonparametric}. Note also that our {real data application} considers a single location, but if {the} data collection presents also a high spatial frequency, it would be possible to proceed with a spatial analysis. {For this, an extension to the infinite-dimensional context of the regression model assumed in \cite{meilan2020a} could be considered, using the results obtained in  \cite{Ruiz2012} or in \cite{Aguilera2017} for spatial functional data as part of the analysis.}}}

		{The numerical analysis carried out in this research was performed with the statistical environment R \citep{RSoft}, using the
			functions supplied with the \texttt{fda.usc} package \citep{fda.usc}. }

\section{Acknowledgments}
Research of A. Meil\'an-Vila and M. Francisco-Fern\'andez has been supported by  {MINECO (Grant MTM2017-82724-R), MICINN (Grant} PID2020-113578RB-I00), and by Xunta de Galicia (Grupos de Referencia Competitiva ED431C-2020-14 and Centro de Investigaci\'on del Sistema Universitario de Galicia ED431G 2019/01), all of them through the ERDF. 
{Research of R. M. Crujeiras has been supported by MICINN (Grant PID2020-116587GB-I00), and by Xunta de Galicia (Grupos de Referencia Competitiva ED431C-2021-24), all of them through the ERDF.} The authors acknowledge the support of Meteogalicia for providing the data used in this research. {The authors thank Prof. Manuel Oviedo, from Universidade da Coru\~na, for his help in the use of the \texttt{fda.usc} package of \texttt{R}. The authors also
	thank  two anonymous referees for numerous
	useful comments that significantly improved this article.}

\begin{appendix}
	
	\section{Proofs of the theoretical results}\label{proofs}
	
	\begin{proof}[Proof of Lemma \ref{C_pro1}]
		For a fixed $\chi$, 	the asymptotic bias and variance of $\hat{m}_{j, h}(\chi)$, for $j=1,2$, can be directly obtained using the asymptotic properties of the Nadaraya--Watson estimator for models with a  {Euclidean} response and a functional covariate \citep{ferraty2007nonparametric}.
		
		Regarding the covariance between  {estimators} $\hat{m}_{1, h}(\chi)$ and $\hat{m}_{2, h}(\chi)$,  {a} decomposition can be obtained following similar  arguments to those used by \cite{friedlander1980} \citep[see also][]{liu1999stat} in the finite-dimensional case. In that setting,  analytic arguments
		about the Taylor expansion of the function $1/z$ around zero were employed and, therefore,  the result can be extended to the functional
		framework. Denoting
		\begin{eqnarray*}
			R_1&=&\frac{1}{nF_{\chi}(h)}\sum_{i=1}^n K(h^{-1}\norm{\mathcal{X}_i-\chi})\sin(\Theta_i),\\
			R_2&=&\frac{1}{nF_{\chi}(h)}\sum_{i=1}^n K(h^{-1}\norm{\mathcal{X}_i-\chi})\cos(\Theta_i)\quad\mbox{and}\\
			S&=&\frac{1}{nF_{\chi}(h)}\sum_{i=1}^n K(h^{-1}\norm{\mathcal{X}_i-\chi}),
		\end{eqnarray*}
		it follows that
		\begin{eqnarray}
			\lefteqn{\mathbb{C}{\rm ov}[\hat{m}_{1, h}(\chi),\hat{m}_{2, h}(\chi)]}\nonumber\\
			&=&	\mathbb{C}{\rm ov}\Bigg[\dfrac{\frac{1}{nF_{\chi}(h)}\sum_{i=1}^n K(h^{-1}\norm{\mathcal{X}_i-\chi})\sin(\Theta_i)}{\frac{1}{nF_{\chi}(h)}\sum_{i=1}^n K(h^{-1}\norm{\mathcal{X}_i-\chi})}\nonumber\\
			&,&\dfrac{\frac{1}{nF_{\chi}(h)}\sum_{i=1}^n K(h^{-1}\norm{\mathcal{X}_i-\chi})\cos(\Theta_i)}{\frac{1}{nF_{\chi}(h)}\sum_{i=1}^n K(h^{-1}\norm{\mathcal{X}_i-\chi})}\Bigg]\nonumber\\
			&=&\dfrac{\mu_{R_1}}{\mu_{S}}\dfrac{\mu_{R_2}}{\mu_{S}}\left[\dfrac{\sigma^2_{S}}{\mu^2_{S}}+\dfrac{\mathbb{C}{\rm ov}(R_1,R_2)}{\sigma_{R_1}\sigma_{R_2}}\dfrac{\sigma_{R_1}}{\mu_{R_1}}\dfrac{\sigma_{R_2}}{\mu_{R_2}}\right.\nonumber\\
			&-&\left.\dfrac{\mathbb{C}{\rm ov}(R_2,S)}{\sigma_{R_2}\sigma_{S}}\dfrac{\sigma_{R_2}}{\mu_{R_2}}\dfrac{\sigma_{S}}{\mu_{S}}-\dfrac{\mathbb{C}{\rm ov}(R_1,S)}{\sigma_{R_1}\sigma_{S}}\dfrac{\sigma_{R_1}}{\mu_{R_1}}\dfrac{\sigma_{S}}{\mu_{S}}\right]+\mathpzc{o}\left[\frac{1}{nF_{\chi}(h)}\right]\nonumber\\
			&=&\dfrac{\mu_{R_1}}{\mu_{S}}\dfrac{\mu_{R_2}}{\mu_{S}}\left[\dfrac{\sigma^2_{S}}{\mu^2_{S}}+\dfrac{\mathbb{C}{\rm ov}(R_1,R_2)}{\mu_{R_1}\mu_{R_2}}-\dfrac{\mathbb{C}{\rm ov}(R_2,S)}{\mu_{R_2}\mu_{S}}-\dfrac{\mathbb{C}{\rm ov}(R_1,S)}{\mu_{R_1}\mu_{S}}\right]\nonumber\\
			&+&\mathpzc{o}\left[\frac{1}{nF_{\chi}(h)}\right] \label{cov_m1_m2},
		\end{eqnarray}
		where $\mu_{R_1}$ and $\sigma^2_{R_1}$, $\mu_{R_2}$ and $\sigma^2_{R_2}$, and $\mu_{S}$ and $\sigma^2_{S}$ denote the expectation and variance of $R_1$, $R_2$ and $S$, respectively. Now,  using Lemma 4 and 5 of \citep{ferraty2007nonparametric}, it can be obtained that
		\begin{eqnarray}
			\mu_{R_1}&\to&m_1(\chi)M_{\chi,1},\label{mur1}	\\
			\mu_{R_2}&\to&m_2(\chi)M_{\chi,1},\label{mur2}	\\
			\mu_{S}&\to&M_{\chi,1}\label{mus},	
		\end{eqnarray}
		and
		\begin{eqnarray}
			\sigma^2_S&=&\dfrac{M_{\chi,2}}{nF_{\chi}(h)}[1+\mathpzc{o}(1)],\label{sigmas}\\
			\mathbb{C}{\rm ov}(R_1,S)&=&m_1(\chi)\dfrac{M_{\chi,2}}{nF_{\chi}(h)}[1+\mathpzc{o}(1)],\label{covr1s}\\
			\mathbb{C}{\rm ov}(R_2,S)&=&m_2(\chi)\dfrac{M_{\chi,2}}{nF_{\chi}(h)}[1+\mathpzc{o}(1)]\label{covr2s}.
		\end{eqnarray}
		
		Moreover,  taking into account the continuity of $m_1$, $m_2$ and $c$, and using the results by \citet[][ {pp. 283}]{ferraty2007nonparametric}, it follows that
		\begin{eqnarray}
			\lefteqn{\mathbb{C}{\rm ov}(R_1,R_2)}\nonumber\\
			&=&\dfrac{1}{nF^2_{\chi}(h)}\bigg\{\mathbb{E}\bigg[\sin\Theta\cos\Theta K^2(h^{-1}\norm{\mathcal{X}-\chi}) \bigg]\nonumber\\
			&-&\mathbb{E}\bigg[\sin\Theta K(h^{-1}\norm{\mathcal{X}-\chi}) \bigg]\mathbb{E}\bigg[\cos\Theta K(h^{-1}\norm{\mathcal{X}-\chi}) \bigg]\bigg\}\nonumber\\
			&=&\dfrac{1}{nF^2_{\chi}(h)}\bigg\{[c(\chi)+m_1(\chi)m_2(\chi)+\mathpzc{o}(1)]\mathbb{E}\bigg[ K^2(h^{-1}\norm{\mathcal{X}-\chi}) \bigg]\nonumber\\
			&-&[m_1(\chi)+\mathpzc{o}(1)]\mathbb{E}\bigg[ K(h^{-1}\norm{\mathcal{X}-\chi}) \bigg][m_2(\chi)+\mathpzc{o}(1)]\mathbb{E}\bigg[ K(h^{-1}\norm{\mathcal{X}-\chi}) \bigg]\bigg\}\nonumber\\
			&=&[c(\chi)+m_1(\chi)m_2(\chi)]\dfrac{M_{\chi,2}}{nF_{\chi}(h)}[1+\mathpzc{o}(1)],\label{covr1r2}
		\end{eqnarray}
		since	
		\begin{eqnarray*}
			\{\mathbb{E}[ K(h^{-1}\norm{\mathcal{X}-\chi})]\}^2 &=&\mathcal{O}[F^2_{\chi}(h)],\\
			\mathbb{E}[ K^2(h^{-1}\norm{\mathcal{X}-\chi})] &=&F_{\chi}(h)\bigg[K^2(1)-\int_0^1(K^2)'(s)\tau_{\chi,h}(s)ds\bigg]\\&=&F_{\chi}(h)M_{\chi,2}.
		\end{eqnarray*} 
		
		Therefore,  {considering \eqref{cov_m1_m2} and using \eqref{mur1}--\eqref{covr1r2}}, it can be directly obtained that
		\begin{eqnarray*}
			\mathbb{C}{\rm ov}[\hat{m}_{1, h}(\chi),\hat{m}_{2, h}(\chi)]\nonumber&=&\dfrac{c(\chi)}{nF_{\chi}(h)}\dfrac{M_{\chi,2}}{M^2_{\chi,1}}+\mathpzc{o}\bigg[\frac{1}{nF_{\chi}(h)}\bigg].
		\end{eqnarray*}
	\end{proof}

	\begin{proof}[Proof of Theorem \ref{C_teoNadaraya--Watson}]
		First, to obtain the bias of $\hat{m}_{h}(\chi)$, given in (\ref{C_est}), the function ${\rm atan2}(\hat{m}_{1,h},\hat{m}_{2,h})$ is expanded in Taylor series around $(m_1,m_2)$,  {where for simplicity, $\hat{m}_{j,h}$ and $m_j$ denote $\hat{m}_{1,h}(\chi)$ and $m(\chi)$, respectively, for $j=1,2,$} to get
		\begin{eqnarray}\label{atan}
			\lefteqn{{\rm atan2}(\hat{m}_{1, h},\hat{m}_{2, h})}\nonumber\\
			&=&{\rm atan2}(m_1,m_2)+\dfrac{m_2}{m_1^2+m_2^2}(\hat{m}_{1, h}-m_1)\nonumber\\
			&-&\dfrac{m_1}{m_1^2+m_2^2}(\hat{m}_{2, h}-m_2)+\dfrac{m_1m_2}{(m_1^2+m_2^2)^2}(\hat{m}_{2, h}-m_2)^2\nonumber\\
			&-&\dfrac{m_1m_2}{(m_1^2+m_2^2)^2}(\hat{m}_{1, h}-m_1)^2-\dfrac{m_1^2-m_2^2}{(m_1^2+m_2^2)^2}(\hat{m}_{1, h}-m_1)(\hat{m}_{2, h}-m_2)\nonumber\\
			&+&\mathcal{O}\left[(\hat{m}_{1, h}-m_1)^3\right]+\mathcal{O}\left[(\hat{m}_{2, h}-m_2)^3\right].
		\end{eqnarray}
		
		Hence, noting that $\ell(\chi)=[m^2_1(\chi)+m_2^2(\chi)]^{1/2}$ and taking expectations in the above expression,  it follows that
		\begin{eqnarray*}
			\lefteqn{{\mathbb{E}}[\hat{m}_{h}(\chi)]-m(\chi)}\\
			&=&\dfrac{m_2(\chi)}{\ell^2(\chi)}{\mathbb{E}}[\hat{m}_{1,h}(\chi)-{m}_{1}(\chi)]- \dfrac{m_1(\chi)}{\ell^2(\chi)}{\mathbb{E}}[\hat{m}_{2,h}(\chi)-{m}_{2}(\chi)]\\
			&+&\dfrac{m_1(\chi)m_2(\chi)}{\ell^4(\chi)}{\mathbb{E}}\{[\hat{m}_{2,h}(\chi)-{m}_{2}(\chi)]^2\}\\
			&-&\dfrac{m_1(\chi)m_2(\chi)}{\ell^4(\chi)}{\mathbb{E}}\{[\hat{m}_{1,h}(\chi)-{m}_{1}(\chi)]^2\}\\
			&-&\dfrac{m_1^2(\chi)-m_2^2(\chi)}{\ell^4(\chi)}{\mathbb{E}}\{[\hat{m}_{1,h}(\chi)-{m}_{1}(\chi)][\hat{m}_{2,h}(\chi)-{m}_{2}(\chi)]\}\\
			&+&\mathcal{O}\{[\hat{m}_{1,h}(\chi)-{m}_{1}(\chi)]^3\}+\mathcal{O}\{[\hat{m}_{2,h}(\chi)-{m}_{2}(\chi)]^3\}.
		\end{eqnarray*}
		
		Noting that ${\mathbb{E}}\left[(\hat{m}_{j,h} -m_j)^2\right]={\mathbb{V}{\rm ar}}(\hat{m}_{j, h})+[\mathbb{E}(\hat{m}_{j, h}-m_j)]^2$, and using the results in Lemma \ref{C_pro1}, it is obtained that 
		\begin{eqnarray}\label{Emm}		
			\lefteqn{{\mathbb{E}}[\hat{m}_{h}(\chi)-m(\chi)]}\nonumber\\
			&=&\dfrac{m_2(\chi)}{\ell^2(\chi)}\varphi_{1,\chi}'(0)\dfrac{M_{\chi,0}}{M_{\chi,1}}h-\dfrac{m_1(\chi)}{\ell^2(\chi)}\varphi_{2,\chi}'(0)\dfrac{M_{\chi,0}}{M_{\chi,1}}h+\mathcal{O}\bigg[\frac{1}{nF_{\chi}(h)}\bigg]+\mathpzc{o}(h)\nonumber\\
			&=&\dfrac{1}{\ell^2(\chi)}\dfrac{M_{\chi,0}}{M_{\chi,1}}h{[m_2(\chi)\varphi_{1,\chi}'(0)-m_1(\chi)\varphi_{2,\chi}'(0)]}+\mathcal{O}\bigg[\frac{1}{nF_{\chi}(h)}\bigg]+\mathpzc{o}(h).
		\end{eqnarray}
		
		Now,  {expression \eqref{Emm} can be further simplify, expanding} the function ${\rm atan2}$ at $[m_1(\mathcal{X}),m_2(\mathcal{X})]$  in Taylor series around $[m_1(\chi),m_2(\chi)]${, as in (\ref{atan})}. Taking conditional expectations in that expansion, one gets that
		\begin{eqnarray}
			\lefteqn{{\mathbb{E}}[{m}(\mathcal{X})-m(\chi)\mid \norm{\mathcal{X}-\chi}=s]}\nonumber\\
			&=&\dfrac{m_2(\chi)}{\ell^2(\chi)}{\mathbb{E}}[m_1(\mathcal{X})-{m}_{1}(\chi)\mid \norm{\mathcal{X}-\chi}=s]\nonumber\\
			&-&\dfrac{m_1(\chi)}{\ell^2(\chi)}{\mathbb{E}}[m_2(\mathcal{X})-{m}_{2}(\chi)\mid \norm{\mathcal{X}-\chi}=s]\nonumber\\
			&+&\dfrac{m_1(\chi)m_2(\chi)}{\ell^4(\chi)}{\mathbb{E}}\{[m_2(\mathcal{X})-{m}_{2}(\chi)]^2\mid \norm{\mathcal{X}-\chi}=s\}\nonumber\\
			&-&\dfrac{m_1(\chi)m_2(\chi)}{\ell^4(\chi)}{\mathbb{E}}\{[m_1(\mathcal{X})-{m}_{1}(\chi)]^2\mid \norm{\mathcal{X}-\chi}=s\}\nonumber\\
			&-&\dfrac{m_1^2(\chi)-m_2^2(\chi)}{\ell^4(\chi)}{\mathbb{E}}\{[m_1(\mathcal{X})-{m}_{1}(\chi)][m_2(\mathcal{X})-{m}_{2}(\chi)]\mid \norm{\mathcal{X}-\chi}=s\}\nonumber\\
			&+&\mathcal{O}\{[m_1(\mathcal{X})-{m}_{1}(\chi)]^3\}+\mathcal{O}\{[m_2(\mathcal{X})-{m}_{2}(\chi)]^3\}.
			\label{EmXmx}
		\end{eqnarray}
		
		Noting that $\varphi_{\chi}(s)=\mathbb{E}[m(\mathcal{X})-m(\chi)\mid \norm{\mathcal{X}-\chi}=s]$    and
		$\varphi_{j,\chi}(s)=\mathbb{E}[m_j(\mathcal{X})-m_j(\chi)\mid \norm{\mathcal{X}-\chi}=s],$ for $j=1,2$,  $\forall s\in\mathbb{R}$, deriving  expression (\ref{EmXmx})  with respect to $s$ and  {replacing the obtained result in} (\ref{Emm}), it follows that
		
		\begin{eqnarray*}		{\mathbb{E}}[\hat{m}_{h}(\chi)-m(\chi)]&=&\varphi_{\chi}'(0)\dfrac{M_{\chi,0}}{M_{\chi,1}}h+\mathcal{O}\bigg[\frac{1}{nF_{\chi}(h)}\bigg]+\mathpzc{o}(h).\end{eqnarray*}

		In order to derive the variance of the estimator $\hat{m}_{h}(\chi)$, the function ${\rm atan2}^2(\hat{m}_{1, h},\hat{m}_{2, h})$ is  expanded  in Taylor series around $(m_1,m_2)$ to obtain that	
		\begin{eqnarray}\label{atan2}
			\lefteqn{{\rm atan2}^2(\hat{m}_{1, h},\hat{m}_{2, h})}\nonumber\\
			&=&{\rm atan2}^2({m}_1,{m}_2)+\dfrac{2{\rm atan2}(m_1,m_2)m_2}{m_1^2+m_2^2}(\hat{m}_{1, h}-{m}_1)\nonumber\\
			&-&\dfrac{2{\rm atan2}(m_1,m_2)m_1}{m_1^2+m_2^2}(\hat{m}_{2, h}-{m}_2)\nonumber\\
			&+&\dfrac{2{\rm atan2}(m_1,m_2)m_1m_2}{(m_1^2+m_2^2)^2}(\hat{m}_{2, h}-{m}_2)^2\nonumber\\
			&-&\dfrac{2{\rm atan2}(m_1,m_2)m_1m_2}{(m_1^2+m_2^2)^2}(\hat{m}_{1, h}-{m}_1)^2\nonumber\\
			&-&\dfrac{2{\rm atan}(m_1,m_2)(m_1^2-m_2^2)}{(m_1^2+m_2^2)^2}(\hat{m}_{1, h}-{m}_1)(\hat{m}_{2, h}-{m}_2)\nonumber\\
			&+&\dfrac{m_1^2}{(m_1^2+m_2^2)^2}(\hat{m}_{2, h}-{m}_2)^2+\dfrac{m_2^2}{(m_1^2+m_2^2)^2}(\hat{m}_{1, h}-{m}_1)^2\nonumber\\
			&-&\dfrac{2m_1m_2}{(m_1^2+m_2^2)^2}(\hat{m}_{1, h}-{m}_1)(\hat{m}_{2, h}-{m}_2)\nonumber\\
			&+&\mathcal{O}\left[(\hat{m}_{1, h}-{m}_1)^3\right]+\mathcal{O}\left[(\hat{m}_{2, h}-{m}_2)^3\right].
		\end{eqnarray}
		
		Taking expectations in the Taylor expansions given in  (\ref{atan}) and (\ref{atan2}), one gets that
		\begin{eqnarray*}
			\lefteqn{{\mathbb{V}{\rm  ar}}[\hat{m}_{h}(\chi)]}\\
			&=&{\rm atan2}^2[{m}_{1,h}(\chi),{m}_{2,h}(\chi)]\\
			&+&\dfrac{2{\rm atan2}[m_1(\chi),m_2(\chi)]m_2(\chi)}{\ell^2(\chi)}{\mathbb{E}}[\hat{m}_{1,h}(\chi)-{m}_{1}(\chi)]\\
			&-&\dfrac{2{\rm atan2}[m_1(\chi),m_2(\chi)]m_1(\chi)}{\ell^2(\chi)}{\mathbb{E}}[\hat{m}_{2,h}(\chi)-{m}_{2}(\chi)]\\
			&+&\dfrac{2{\rm atan2}[m_1(\chi),m_2(\chi)]m_1(\chi)m_2(\chi)}{\ell^4(\chi)}{\mathbb{E}}\{[\hat{m}_{2,h}(\chi)-{m}_{2}(\chi)]^2\}\\
			&-&\dfrac{2{\rm atan2}[m_1(\chi),m_2(\chi)]m_1(\chi)m_2(\chi)}{\ell^4(\chi)}{\mathbb{E}}\{[\hat{m}_{1,h}(\chi)-{m}_{1}(\chi)]^2\}\\
			&-&\dfrac{2{\rm atan}[m_1(\chi),m_2(\chi)][m_1^2(\chi)-m_2^2(\chi)]}{\ell^4(\chi)}{\mathbb{E}}\{[\hat{m}_{1,h}(\chi)-{m}_{1}(\chi)][\hat{m}_{2,h}(\chi)-{m}_{2}(\chi)]\}\\
			&+&\dfrac{m_1^2(\chi)}{\ell^4(\chi)}{\mathbb{E}}\{[\hat{m}_{2,h}(\chi)-{m}_{2}(\chi)]^2\}+\dfrac{m_2^2(\chi)}{\ell^4(\chi)}{\mathbb{E}}\{[\hat{m}_{1,h}(\chi)-{m}_{1}(\chi)]^2\}\\
			&-&\dfrac{2m_1(\chi)m_2(\chi)}{\ell^4(\chi)}{\mathbb{E}}\{[\hat{m}_{1,h}(\chi)-{m}_{1}(\chi)][\hat{m}_{2,h}(\chi)-{m}_{2}(\chi)]\}\\
			&-&\bigg({\rm atan2}[m_1(\chi),m_2(\chi)]+\dfrac{m_2(\chi)}{\ell^2(\chi)}{\mathbb{E}}[\hat{m}_{1,h}(\chi)-{m}_{1}(\chi)]\\
			&-&\dfrac{m_1(\chi)}{\ell^2(\chi)}{\mathbb{E}}[\hat{m}_{2,h}(\chi)-{m}_{2}(\chi)]+\dfrac{m_1(\chi)m_2(\chi)}{\ell^4(\chi)}{\mathbb{E}}\{[\hat{m}_{2,h}(\chi)-{m}_{2}(\chi)]^2\}\\
			&-&\dfrac{m_1(\chi)m_2(\chi)}{\ell^4(\chi)}{\mathbb{E}}\{[\hat{m}_{1,h}(\chi)-{m}_{1}(\chi)]^2\}\nonumber\\
			&&- \dfrac{m_1^2(\chi)-m_2^2(\chi)}{\ell^4(\chi)}{\mathbb{E}}\{[\hat{m}_{1,h}(\chi)-{m}_{1}(\chi)][\hat{m}_{2,h}(\chi)-{m}_{2}(\chi)]\}\\
			&+&\mathcal{O}\{[\hat{m}_{1,h}(\chi)-{m}_{1}(\chi)]^3\}+\mathcal{O}\{[\hat{m}_{2,h}(\chi)-{m}_{2}(\chi)]^3\}\bigg)^2\\
			&+&\mathcal{O}\{[\hat{m}_{1,h}(\chi)-{m}_{1}(\chi)]^3\}+\mathcal{O}\{[\hat{m}_{2,h}(\chi)-{m}_{2}(\chi)]^3\}.\end{eqnarray*}
		
		By straightforward calculations, it can be obtained that
		\begin{eqnarray*}
			\lefteqn{\mathbb{V}{\rm ar}[\hat{m}_{h}(\chi)]}\\
			&=&\dfrac{m_1^2(\chi)}{\ell^4(\chi)}{\mathbb{E}}\{[\hat{m}_{2,h}(\chi)-{m}_{2}(\chi)]^2\}+\dfrac{m_2^2(\chi)}{\ell^4(\chi)}{\mathbb{E}}\{[\hat{m}_{1,h}(\chi)-{m}_{1}(\chi)]^2\}\\
			&-&\dfrac{2m_1(\chi)m_2(\chi)}{\ell^4(\chi)}{\mathbb{E}}\{[\hat{m}_{1,h}(\chi)-{m}_{1}(\chi)][\hat{m}_{2,h}(\chi)-{m}_{2}(\chi)]\}\\
			&-&\dfrac{m^2_2(\chi)}{\ell^4(\chi)}\{{\mathbb{E}}[\hat{m}_{1,h}(\chi)-{m}_{1}(\chi)]\}^2 - \dfrac{m^2_1(\chi)}{\ell^4(\chi)}\{{\mathbb{E}}[\hat{m}_{2,h}(\chi)-{m}_{2}(\chi)]\}^2\\
			&+&\dfrac{2m_1(\chi)m_2(\chi)}{\ell^4(\chi)}{\mathbb{E}}[\hat{m}_{1,h}(\chi)-{m}_{1}(\chi)]{\mathbb{E}}[\hat{m}_{2,h}(\chi)-{m}_{2}(\chi)]\\
			&+&\mathcal{O}\{[\hat{m}_{1,h}(\chi)-{m}_{1}(\chi)]^3\}+\mathcal{O}\{[\hat{m}_{2,h}(\chi)-{m}_{2}(\chi)]^3\}.
		\end{eqnarray*}

		So,  noting that ${\mathbb{E}}\left[(\hat{m}_{j,h} -m_j)^2\right]={\mathbb{V}{\rm ar}}(\hat{m}_{j, h})+[\mathbb{E}(\hat{m}_{j, h}-m_j)]^2$, it can be obtained that the conditional variance is:
		\begin{eqnarray*}
			\lefteqn{\mathbb{V}{\rm ar}[\hat{m}_{h}(\chi)]}\\
			&=&\dfrac{m_1^2(\chi)}{\ell^4(\chi)}{\mathbb{V}{\rm ar}}[\hat{m}_{2, h}(\chi)]+ \dfrac{m_2^2(\chi)}{\ell^4(\chi)}{\mathbb{V}{\rm ar}}[\hat{m}_{1, h}(\chi)]\\
			&-&\dfrac{2m_1(\chi)m_2(\chi)}{\ell^4(\chi)}\mathbb{C}{\rm ov}[\hat{m}_{1, h}(\chi),\hat{m}_{2, h}(\chi)]\\
			&+&\mathcal{O}\left\{[\hat{m}_{1, h}(\chi)-{m}_1(\chi)]^3\right\}+\mathcal{O}\left\{[\hat{m}_{2, h}(\chi)-{m}_2(\chi)]^3\right\}.
		\end{eqnarray*}

		Therefore, using Lemma \ref{C_pro1}, one gets that
		\begin{eqnarray*}{\mathbb{V}{\rm ar}}[\hat{m}_{h}(\chi)]&=&\dfrac{m_1^2(\chi)}{\ell^4(\chi)}\dfrac{s_2^2(\chi)}{nF_{\chi}(h)}\dfrac{M_{\chi,2}}{M^2_{\chi,1}}+ \dfrac{m_2^2(\chi)}{\ell^4(\chi)}\dfrac{s_1^2(\chi)}{nF_{\chi}(h)}\dfrac{M_{\chi,2}}{M^2_{\chi,1}}\\
			&-&\dfrac{2m_1(\chi)m_2(\chi)}{\ell^4(\chi)}\dfrac{c(\chi)}{nF_{\chi}(h)}\dfrac{M_{\chi,2}}{M^2_{\chi,1}}+\mathpzc{o}\bigg[\frac{1}{nF_{\chi}(h)}\bigg].\end{eqnarray*}

		Considering  {equations (\ref{m1m2})--(\ref{c})}, and taking into account that  $f_1^2(\chi)+f_2^2(\chi)=1$, it follows that
		\begin{eqnarray*}
			\label{l2sigma12}
			\lefteqn{m_1^2(\chi)s_2^2(\chi)+m_2^2(\chi)s_1^2(\chi)-2m_1(\chi)m_2(\chi)c(\chi)}\nonumber\\
			&=&f_1^2(\chi)f_2^2(\chi)\ell^2(\chi)\sigma^2_2(\chi)-2f_2(\chi)f_1^3(\chi)\ell^2(\chi)\sigma_{12}(\chi)\\
			&+&f_1^4(\chi)\ell^2(\chi)\sigma^2_1(\chi)+f_2^2(\chi)f_1^2(\chi)\ell^2(\chi)\sigma^2_2(\chi)\\
			&+&2f_2^3(\chi)f_1(\chi)\ell^2(\chi)\sigma_{12}(\chi)+f_2^4(\chi)\ell^2(\chi)\sigma^2_1(\chi)\\
			&-&2f_1^2(\chi)f_2^2(\chi)\ell^2(\chi)\sigma^2_2(\chi)+2f_1^3(\chi)f_2(\chi)\ell^2(\chi)\sigma_{12}(\chi)\\
			&-&2f_2^3(\chi)f_1(\chi)\ell^2(\chi)\sigma_{12}(\chi)+2f^2_1(\chi)f^2_2(\chi)\ell^2(\chi)\sigma^2_1(\chi)\\
			&=&\ell^2(\chi)\sigma^2_1(\chi).
		\end{eqnarray*}
		
		Therefore,	
		\begin{eqnarray*}{\mathbb{V}{\rm ar}}[\hat{m}_{h}(\chi)]&=&\dfrac{1}{nF_{\chi}(h)} \dfrac{\sigma^2_1(\chi)}{\ell^2(\chi)}\dfrac{M_{\chi,2}}{M^2_{\chi,1}}+\mathpzc{o}\bigg[\frac{1}{nF_{\chi}(h)}\bigg].\label{var_m_p0}\end{eqnarray*}
	\end{proof}

	\begin{proof}[Proof of Theorem \ref{eq:theorem2}]
		In order to derive the asymptotic distribution of $\hat{m}_h$, given in (\ref{C_est}), we  compute the asymptotic distribution of $\hat{M}=\hat{m}_{1,h}/\hat{m}_{2,h}$  and  apply Theorem~A of \cite{serfling2009approximation}. First,  {using similar arguments to those in} Lemma 6 of \cite{ferraty2007nonparametric}, it can be obtained that
		$$\hat{M}(\chi)-\mathbb{E}[\hat{M}(\chi)]=\dfrac{\hat{m}_{1,h}(\chi)}{\hat{m}_{2,h}(\chi)}-\dfrac{\mathbb{E}[\hat{m}_{1,h}(\chi)]}{\mathbb{E}[\hat{m}_{2,h}(\chi)]}+\mathpzc{o}\bigg[\frac{1}{\sqrt{nF_{\chi}(h)}}\bigg].$$
		
		
		However, the following decomposition holds: 
		\begin{eqnarray*}
			\lefteqn{\dfrac{\hat{m}_{1,h}(\chi)}{\hat{m}_{2,h}(\chi)}-\dfrac{\mathbb{E}[\hat{m}_{1,h}(\chi)]}{\mathbb{E}[\hat{m}_{2,h}(\chi)]}}\\
			&=&\dfrac{\{\hat{m}_{1,h}(\chi)-\mathbb{E}[\hat{m}_{1,h}(\chi)]\}\mathbb{E}[\hat{m}_{2,h}(\chi)]}{\hat{m}_{2,h}(\chi)\mathbb{E}[\hat{m}_{2,h}(\chi)]}+\dfrac{\{\mathbb{E}[\hat{m}_{2,h}(\chi)]-\hat{m}_{2,h}(\chi)\}\mathbb{E}[\hat{m}_{1,h}(\chi)]}{\hat{m}_{2,h}(\chi)\mathbb{E}[\hat{m}_{2,h}(\chi)]}
		\end{eqnarray*}
		
		Therefore, $\hat{M}(\chi)-\mathbb{E}[\hat{M}(\chi)]$ has the same asymptotic distribution  {as} 
		\begin{eqnarray*}
			N(\chi)&=&\dfrac{\{\hat{m}_{1,h}(\chi)-\mathbb{E}[\hat{m}_{1,h}(\chi)]\}\mathbb{E}[\hat{m}_{2,h}(\chi)]}{\hat{m}_{2,h}(\chi)\mathbb{E}[\hat{m}_{2,h}(\chi)]}+\dfrac{\{\mathbb{E}[\hat{m}_{2,h}(\chi)]-\hat{m}_{2,h}(\chi)\}\mathbb{E}[\hat{m}_{1,h}(\chi)]}{\hat{m}_{2,h}(\chi)\mathbb{E}[\hat{m}_{2,h}(\chi)]}.
		\end{eqnarray*}
		
		Note that $$[\hat{m}_{1,h}(\chi)/\hat{m}_{2,h}(\chi)]-\{\mathbb{E}[\hat{m}_{1,h}(\chi)]/\mathbb{E}[\hat{m}_{2,h}(\chi)]\}$$ can be expressed as an array of independent centered random variables and, consequently, the Central Limit Theorem can be applied. Moreover, using the Theorem of Slutsky, it can be obtained that $N(\chi)$ also follows a normal distribution.
		
		{The asymptotic bias and variance of $\hat{M}$ could be computed by expanding the function $\hat{m}_{1,h}/\hat{m}_{2,h}$  in Taylor series around $m_1/m_2$ and using similar steps to those}
		employed to derive $	{\mathbb{E}}[\hat{m}_{h}(\chi)]$ and $\mathbb{V}{\rm ar}[\hat{m}_{h}(\chi)]$.  {Therefore, denoting by $\bm{b}_h$ and  $v$ the leading terms of the asymptotic bias and variance of $\hat{M}$, respectively,} it follows that, as $n\to\infty,$
		$$\dfrac{\hat{M}(\chi)-m_1(\chi)/m_2(\chi)-\bm{b}_h}{v^{1/2}(\chi)}\xrightarrow[]{\mathcal{L}}N(0,1).$$
		Finally,  applying Theorem A of \cite{serfling2009approximation} and Theorem \ref{C_teoNadaraya--Watson}, it can be obtained that {, as $n\to\infty,$}
		$$\sqrt{n{F}_{\chi}(h)}\dfrac{\ell(\chi)M_{\chi,1}}{\sqrt{\sigma^2_1(\chi)M_{\chi,2}}}[\hat{m}_{h}(\chi)-m(\chi)-\bm{B}_h]\xrightarrow[]{\mathcal{L}}N(0,1).$$
	\end{proof}

	\begin{proof}[Proof of Corollary \ref{eq:cor1}] This result can be obtained by combining Theorem \ref{eq:theorem2} and assumption (C7).
	\end{proof}
	
	\begin{proof}[Proof of Corollary \ref{eq:cor2}] Using the Glivenko--Cantelli Theorem,
		it follows that
		$$\dfrac{\hat{F}_{\chi}(h)}{{F}_{\chi}(h)}\xrightarrow[]{\mathbb{P}}1.$$
		
		Therefore,  {using this result, assumption (C7), and the consistency of estimators {$\hat{\sigma}_1^2$} and $\hat{l}$, as well as} Theorem \ref{eq:theorem2}, Corollary \ref{eq:cor2} can be directly obtained. 
	\end{proof}
	
\end{appendix}

\bibliographystyle{apalike}
\bibliography{bibibnew2}

\end{document}